\documentclass[10pt, conference, letterpaper]{IEEEtran}

\usepackage{listings}
\usepackage{amsfonts}

\usepackage{multicol}

\usepackage{blindtext}
\usepackage{amsfonts}
\usepackage{amsmath, amsthm, amssymb}
\usepackage{times}
\usepackage{url}
\usepackage{verbatim}
\usepackage{array}
\usepackage{mathrsfs}
\usepackage{enumitem}
\usepackage{textcomp}

\newcommand{\subparagraph}{}
\usepackage[compact]{titlesec}
\usepackage[ruled, lined, linesnumbered]{algorithm2e}
\SetAlFnt{\small}

\SetCommentSty{mycommfont}

\usepackage{amsfonts}
\usepackage{bbm}
\usepackage{soul}
\usepackage{color}
\usepackage{graphics}
\usepackage{graphicx,epsfig}
\usepackage{lipsum,graphicx,subcaption}
\graphicspath{{figures/}}
\usepackage{caption}
\usepackage{subcaption}

\usepackage{cleveref}
\usepackage{wrapfig}
\usepackage{url}
\usepackage{url}
\usepackage{amsmath,amssymb,epsfig,amsthm,psfrag,epsf, multirow,wrapfig, graphicx}
\usepackage[english]{babel}
\usepackage{epstopdf}
\usepackage{float}
\usepackage{color}
\usepackage[table,xcdraw]{xcolor}
\usepackage{amsmath}
\usepackage{amsfonts}
\usepackage{amssymb}
\usepackage{mathtools}

\usepackage{tabularx}
\usepackage{adjustbox}
\usepackage{physics}

\usepackage{amsfonts}
\usepackage{amssymb}
\usepackage{caption}
\usepackage{subcaption}
\usepackage{tabularx}
\usepackage{comment}
\usepackage[mathscr]{euscript}
\usepackage{amsbsy}
\usepackage{stmaryrd}
\usepackage{mathtools}
\usepackage[utf8]{inputenc}
\graphicspath{{figs/}}

\usepackage{booktabs}
\usepackage{graphicx,epsfig}
\usepackage{graphics}
\usepackage{multirow}
\usepackage{rotating}
\usepackage{caption}
\usepackage{subcaption}
\usepackage{amsmath}
\usepackage{amssymb}
\usepackage{makecell}
\DeclareMathOperator*{\argmax}{arg\,max}
\DeclareMathOperator*{\argmin}{arg\,min}
\usepackage{mathtools}
\usepackage{comment}
\usepackage{multirow}

\usepackage{footnote}
\usepackage{tablefootnote}
\usepackage{scalerel}

\usepackage{tabularx}
\usepackage{mdframed}
\usepackage{lipsum}
\makeatletter
\newcommand*{\rom}[1]{\expandafter\@slowromancap\romannumeral #1@}
\makeatother
\usepackage{nccmath}

\usepackage{relsize}

\newtheorem{theorem}{Theorem}

\newtheorem{corollary}{Corollary}

\newtheorem{lemma}{Lemma}

\renewcommand\IEEEkeywordsname{Keywords}

\DeclareMathSymbol{\shortminus}{\mathbin}{AMSa}{"39}

\allowdisplaybreaks
\begin{document}
\def\eg{\mbox{\em e.g.}, }


\title{Capacity Achieving by Diagonal Permutation for MU-MIMO channels}

\author{
\IEEEauthorblockN{Zhibin Zou, and Aveek Dutta}
    \IEEEauthorblockA{Department of Electrical and Computer Engineering\\
    University at Albany SUNY, Albany, NY 12222 USA\\
    \{zzou2, adutta\}@albany.edu}
    \vspace{-5ex}

}
    
\maketitle

\begin{abstract}
Dirty Paper Coding (DPC) is considered as the optimal precoding which achieves capacity for the Gaussian Multiple-Input Multiple-Output (MIMO) broadcast channel (BC). However, to find the optimal precoding order, it needs to repeat $N!$ times for $N$ users as there are $N!$ possible precoding orders. This extremely high complexity limits its practical use in modern wireless networks. In this paper, we show the equivalence of DPC and the recently proposed Higher Order Mercer's Theorem (HOGMT) precoding~\cite{Zou2022ICC, ZouTCOM23} in 2-D (spatial) case, which provides an alternate implementation for DPC. Furthermore, we show that the proposed implementation method is linear over the permutation operator when permuting over multi-user channels. Therefore, we present a low complexity algorithm that optimizes the precoding order for DPC with beamforming, eliminating repeated computation of DPC for each precoding order. Simulations show that our method can achieve the same result as conventional DPC with ${\approx}20$ dB lower complexity for $N {=} 5$ users.
\end{abstract}
\vspace{-10pt}
\vspace{-5pt}
\section{Introduction}
\label{sec:intro}


Precoding 
is a very well investigated area, which can cancel interference if the CSI is available at the transmitter~\cite{Vu2007Precoding}. DPC is a non-linear precoding that achieves 
optimal interference-free transmission by subtracting the potential interference at the transmitter~\cite{CostaDPC1983}, which is well investigated for MU-MIMO channels
~\cite{Cho2010MIMObook}.
In 
multi-user information theory literature, the downlink MU-MIMO channel is modeled as MIMO Gaussian broadcast channel (BC)~\cite{Ma2016MUMIMO}, where the sum-rate capacity grows linearly with the number of spatial-domain degrees of freedom~\cite{Caire2003capacity}. DPC is proven to achieve capacity for MIMO BC channels~\cite{Vishwanath2003Dual,LeeDPC2007,Weingarten2006Capacity}. However, practical implementation of DPC has the great challenge of \textit{very high computational complexity}. At the same time, the power allocation problem is studied in the beamforming literature from linear methods~\cite{Rashid1998Bm,
Schubert2004Bm} to nonlinear DPC~\cite{Schubert2005Bm}.
However, these approaches only solve the problem for a \textit{fixed DPC precoding order}. An inherent problem with DPC is that for every order, the interference coupling matrix has a different structure. Thus finding the optimum precoding order remains a combinatorial problem that is prohibitive, even for moderate numbers of users. A low complexity but sub-optimal method to achieve this 
has been shown in~\cite{Michel2007Order}.  

HOGMT precoding~\cite{Zou2022ICC, ZouTCOM23} is the first method, 
which is capable of cancelling spatial, temporal and joint spatio-temporal interference in multi-user non-stationary channels. This is  achieved 
by transmitting signals on independent flat-fading subchannels (eigenfunctions) in an eigen-domain. As a joint spatio-temporal precoding method for multi-user non-stationary channels, HOGMT generally 
analyzes a 4-D channel tensors. However, if time dimension at the transmitter and the receiver are both collapsed, as it would be LTI channels, it will operate on a 2-D MU-MIMO channel matrix to cancel spatial interference only, which is exactly the same as in DPC for MU-MIMO channels~\cite[Chapter~13]{Cho2010MIMObook}.

In this paper, we 
prove the equivalence between DPC and 2-D HOGMT, since both ensure interference-free communication. 2-D HOGMT precoding is implemented by SVD decomposition~\cite{Zou2022ICC}, which is a linear process. Therefore, the equivalence provides an alternate linear implementation for DPC. Furthermore, we show that the SVD decomposition of a permuted matrix can be obtained by directly permuting the decomposed components (Lemma~\ref{lemma:svd_per}). This property does not hold for the LQ decomposition, commonly used in implementing DPC, because the decomposed triangular matrix cannot preserve its structure after permutation. This difference suggests that conventional method based on LQ decomposition needs to repeat the DPC for each permutation of the channel matrix (or the precoding order), while the proposed alternate method 
requires \textit{only one DPC computation} for any arbitrary order, followed by \textit{permuting the decomposed components} to find the optimal order, avoiding unnecessary iterations. The contributions of this paper are summarized as follows:
\vspace{-2pt}
\begin{itemize}
    \item We show the equivalence between DPC and 2-D HOGMT precoding with effective channel gains, and give an alternate implementation for DPC.
    \item We give a general beamforming optimization method 
    by designing a diagonal matrix 
    according to the target criteria under certain constraints.
    \item We show the difference between SVD and LQ decomposition under the permutation operation and demonstrate the alternative implementation of the DPC is able to optimize precoding order by the diagonal permutation. 
    \item We show the convergence of the proposed method and validate the equivalence by the simulation. 
\end{itemize}
\vspace{-5pt}
\section{Background \& Preliminaries}
\label{sec:pre}
\subsection{DPC for MIMO Broadcast Channels}
Consider a BC channel with $N$ transmit antennas and $N$ single-antenna users (MU-MISO), $\mathbf{H} \in \mathbb{C}^{N\times N}$, where the received signal $\mathbf{y} \in \mathbb{C}^{N \times 1}$ is given by,
\begin{equation}
\label{eq:y}
    \mathbf{y} = \mathbf{H}\mathbf{x} + \mathbf{v}
\end{equation}
where, $\mathbf{x} \in \mathbb{C}^{N \times 1}$ is the precoded signal and $\mathbf{v} \in \mathbb{C}^{N \times 1}$ is AWGN. Computationally, DPC performs LQ decomposition followed by a series of Gram-Schmidt processes~\cite{Cho2010MIMObook}. The channel matrix, $\mathbf{H}$ is decomposed 
as $\mathbf{H} {=} \mathbf{L}\mathbf{Q}$
where, $\mathbf{L} \in \mathbb{C}^{N \times N}$ and $\mathbf{Q}^{N \times N}$ is a triangular and unitary matrix, respectively. Let $\Tilde{\mathbf{x}}{=}[\begin{array}{lll}x_1 , \ldots , x_N\end{array}]^T$ denote the precoded signal for  $\mathbf{s}{=}[\begin{array}{lll}s_1 , \hdots , s_N\end{array}]^T$ to cancel the effect of $\mathbf{L}$. By transmitting $\mathbf{x} {=} \mathbf{Q}^H \Tilde{\mathbf{x}}$, the effect of $\mathbf{Q}$ is cancelled and~\eqref{eq:y} is rewritten as, 
\begin{align}
\mathbf{y} & = \mathbf{H}\mathbf{x} + \mathbf{v} = \mathbf{L} \mathbf{Q} \mathbf{Q}^H \Tilde{\mathbf{x}}+\mathbf{v} \nonumber \\
&=\left[\begin{array}{llll}
l_{11} & 0 & \cdots & 0\\
l_{21} & l_{22} & \cdots & 0\\
\vdots & \vdots & \ddots & \vdots\\
l_{N1} & l_{N2} & \cdots & l_{NN}
\end{array}\right]\left[\begin{array}{c}
\Tilde{x}_1 \\
\Tilde{x}_2 \\
\vdots \\ 
\Tilde{x}_N
\end{array}\right]+\left[\begin{array}{c}
v_1 \\
v_2 \\
\vdots \\
v_N
\end{array}\right]
\label{eq:y1}
\end{align}

Therefore, for $n^\text{th}$ user, there is no interference from users $n'{>}n$ and the interference from users $n'{<}n$ is cancelled by the Gram–Schmidt process as, 
\begin{equation}
\Tilde{x}_n=s_n-\sum_{n'=1}^{n-1} \frac{l_{nn'}}{l_{nn}}\Tilde{x}_{n'} \quad\text{where}, \quad \Tilde{x}_1 = s_1
\label{eq:x1}
\end{equation}
Substituting~\eqref{eq:x1} in~\eqref{eq:y1}, the received signal is given by, 
\begin{align}
    \mathbf{y}&= \mathbf{D_L} \mathbf{s} + \mathbf{v}
\label{eq:y2}
\end{align}
where, $\mathbf{D_L}{=}\operatorname{diag}(\mathbf{L})$
and $l_{nn}$ is the channel gain for user $n$. 

\noindent
\textbf{Multi-antenna users case (MU-MIMO)}: For multi-antenna user case, each row in~\eqref{eq:y1} corresponds to one antenna instead of one user and then each user would incorporate multiple rows as well as multiple elements $\Tilde{\mathbf{x}}_n$ in~\eqref{eq:x1}. For notational simplicity, we use the expression of the single-antennas user case as it does not affect the underlying theory in this paper.

\subsection{HOGMT Precoding}
HOGMT precoding~\cite{Zou2022ICC} cancels the spatial, temporal and joint spatio-temporal interference in a 4-D double-selective channel, modeled as in~\cite{Almers2007ChannelSurvey},
\begin{align}
\label{eq:H_t_tau}
\mathbf{H}(t,\tau) = \begin{bmatrix}
 h_{1,1}(t,\tau) &\cdots & h_{1,u'}(t,\tau)  \\
 \vdots& \ddots &  \\
h_{u,1}(t,\tau) &  & h_{u,u'}(t,\tau)
\end{bmatrix}
\end{align}
where, $h_{u,u'}(t, t')$ is the multi-user time-varying impulse response. Then the received signal is given by, 
\begin{align}
    \label{eq:MU2}
    & r(u,t) = \iint k_H(u,t;u',t') s(u',t')~du'~dt' + v(u,t)
\end{align}
where, $v(u,t)$ is AWGN, $s(u,t)$ is the data symbol and $k_{u,u'} (t, t') {=} h_{u,u'}(t, t{-}t')$ is the 4-D channel kernel~\cite{Matz2005NS,MATZ20111}.

HOGMT decomposition is the first method to decompose a 4-D channel kernels as follows,
\begin{align}
\label{eq:thm2_decomp}
&k_H(u,t;u',t') = \sum\nolimits_{n{=1}}^N  \sigma_n \psi_n(u,t) \phi_n(u',t')
\end{align}
with orthonormal properties as in \eqref{eq:thm2_decomp_pty},
\begin{align}
\label{eq:thm2_decomp_pty}
\begin{aligned}
& {\langle} \psi_n(u{,}t) {,} \psi_{n'}^*(u{,}t) {\rangle} = \delta_{nn'} \\
& {\langle} \phi_n(u{,}t) {,} \phi_{n'}^*(u{,}t) {\rangle} = \delta_{nn'}
\end{aligned}
\end{align}
Both \eqref{eq:thm2_decomp} and~\eqref{eq:thm2_decomp_pty} show that the 4-D channel kernel is decomposed into jointly orthogonal subchannels (eigenfunctions). Then the precoded signal $x(u,t)$ based on HOGMT is derived by combining the jointly orthogonal eigenfunctions with the desired coefficients $x_n$ as,
\begin{align}
\label{eq:x}
    x(u{,}t) {=} \sum_{n{=}1}^N x_n \phi_n^*(u{,}t) 
    \text{ where, } 
    x_n {=} \frac{{\langle} s(u{,}t){,} \psi_n(u{,}t) {\rangle}}{\sigma_n}
\end{align}
Transmitting $x(u,t)$ over the channel, the received signal is directly the combination of data signal and noise without complementary post-coding step as $r(u,t) = s(u,t) + v(u,t)$. It shows that HOGMT precoding can achieves interference-free communication for multi-user non-stationary channels.

\section{Equivalence of DPC and HOGMT precoding}
\label{sec:intro}

DPC achieves capacity for MU-MIMO BC channels but is a non-linear precoding with impractical complexity. On the contrary, HOGMT achieves the same interference-free communication for multidimensional non-stationary channels and has a linear implementation. If there exists equivalence between them, then we can use it as an alternate implementation for DPC for practical system implementation. 
\begin{theorem}
(DPC and 2-D HOGMT precoding with effective channel gains are mathematically equivalent)

Given a channel matrix $\mathbf{H}$ with entries $h(u,u')$ and data symbols $\mathbf{s} {=} [s_1,{...},s_N]^T$, the 2-D HOGMT precoded signal is,

\begin{align}
x(u) = \sum_n^{N} x_n \phi_n^*(u) \text{ where, } x_n =  \frac{\langle s(u), \psi_n(u) \rangle}{\sigma_n}
\label{eq:2-Dhogmt}
\end{align}
\text{where} $\sigma_n$, $\phi_n(u)$ and $\psi_n(u)$ are given by 2-D HOGMT decomposition as $h(u,u') {=} \sum_n^N \sigma_n \psi_n(u) \phi_n(u')$~\cite{Zou2022ICC}. 

Then the DPC precoded signal is given by 
\begin{align}
x(u) = \sum_n^{N} x_n \phi_n^*(u) \text{ where, } x_n =  \frac{\langle l(u), s(u), \psi_n(u) \rangle}{\sigma_n}
    \label{eq:x_u}
\end{align}
where, $l(u)$ is the continuous 
diagonal element of $\mathbf{D_L}$ in~\eqref{eq:y2}.
\label{thm:thm1}
\end{theorem}

\begin{proof}
Let $\mathbf{x} {=} \mathbf{W}\mathbf{s}$, where $\mathbf{W}$ is the precoding matrix. Then we can write, 
\begin{equation}
    \mathbf{y} = \mathbf{H}\mathbf{x} + \mathbf{v} = \mathbf{L}\mathbf{Q} \mathbf{W}\mathbf{s} +\mathbf{v}
    \label{eq:y3}
\end{equation}

Substituting~\eqref{eq:y2} in~\eqref{eq:y3}, we have $\mathbf{W} {=} \mathbf{Q}^H\mathbf{L}^{-1}\mathbf{D_L}$. Decomposing $\mathbf{L}$ by SVD as $\mathbf{L} = \mathbf{U_L}\mathbf{\Sigma} \mathbf{V}_L^H$ we get,
 \begin{equation}
     \mathbf{W} =  \mathbf{Q}^H \mathbf{V_L}\mathbf{\Sigma_L}^{-1}\mathbf{U}_\mathbf{L}^{H} \mathbf{D_L} \label{prof_1}
 \end{equation}
Meanwhile, the SVD of $\mathbf{H}$ can be also represented by the SVD of $\mathbf{L}$ as,
\begin{align}
    \mathbf{H} & = \mathbf{L}\mathbf{Q} = (\mathbf{U_L}\mathbf{\Sigma_L}\mathbf{V_L}^{H})\mathbf{Q} \nonumber\\
    & = \mathbf{U_L}\mathbf{\Sigma_L}(\mathbf{Q}^H \mathbf{V_L})^{H} \equiv \mathbf{U} \mathbf{\Sigma} \mathbf{V}^{H} 
\end{align}

Therefore we have the following equivalence,
\begin{equation}
\mathbf{U} \equiv \mathbf{U_L}, \  \mathbf{\Sigma} \equiv \mathbf{\Sigma_L} \text{ and } \mathbf{V} \equiv \mathbf{Q}^H \mathbf{V_L}
\label{eq:eqv}
\end{equation}
Substituting \eqref{eq:eqv} in \eqref{prof_1} and noting that $\mathbf{Q}$ is unitary,
\begin{figure}
\centering
\includegraphics[width=\linewidth]{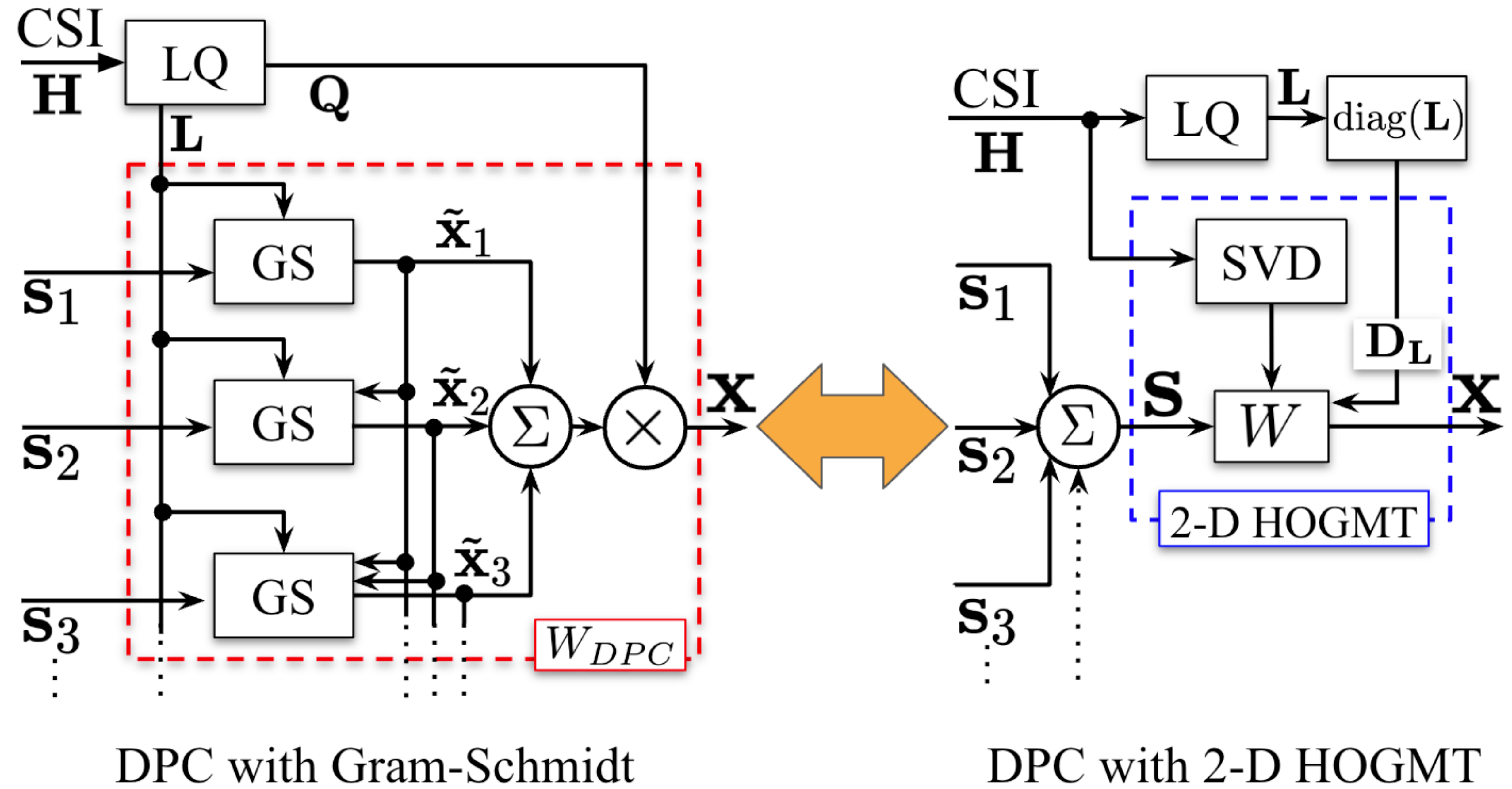}
\caption{A low-complexity implementation of DPC}
  \label{fig:eq}
\end{figure}
\begin{align}
\mathbf{W} =  \mathbf{V} \mathbf{\Sigma}^{-1}\mathbf{U}^{H} \mathbf{D_L} \label{eq:W}
\end{align}
Then the transmitted symbol $\mathbf{x}$ is given by,
\begin{align}
\label{eq:xvec}
    &\mathbf{x} {=} \mathbf{Ws} {=} \mathbf{V} \mathbf{\Sigma}^{-1}\mathbf{U}^{H} \mathbf{D_L} \mathbf{s}
\end{align}
Note that the precoded symbol for user $u$ is the $u^\text{th}$ row of $\mathbf{x}$. By expanding~\eqref{eq:xvec}, we have
\begin{align}
    &x_u {=} \sum_n^N v_{un} \underbrace{\sigma_n^{-1} \sum_u^N l_u s_u \mu_{un}^*}_{x_n}
\label{eq:x_row}
\end{align}
where, $v_{un}$ and $\mu_{un}$ are the elements of $\mathbf{V}$ and $\mathbf{U}$ respectively and $\sigma_n$ is $n^\text{th}$ diagonal element of $\mathbf{\Sigma}$. 

Now, rewriting~\eqref{eq:x_row} using two arbitrary continuous-time complex functions, $\phi_n^*(u)$ and $\psi_n^*(u)$ we get,
\begin{align}
x(u) = \sum_n^{N} x_n \phi_n^*(u) \text{ where, } x_n =  \frac{\langle l(u), s(u), \psi_n(u) \rangle}{\sigma_n}
\label{eq:dpc1}
\end{align}
where, $x(u)$, $s(u)$, and $l(u)$ is the continuous form of $\mathbf{x}$, $\mathbf{s}$ and $\{l_u\}$, respectively.

Meanwhile, the continuous form of SVD of $\mathbf{H}$, yields the two eigenfunctions, $\phi$ and $\psi$ according to the the 2-D HOGMT decomposition in~\eqref{eq:kuu} by collapsing time dimension in~\eqref{eq:thm2_decomp}, 
\begin{equation}
    k(u,u') = \Sigma_n \sigma_n \phi_n(u) \psi(u')
    \label{eq:kuu}
\end{equation}
Therefore, we have the 2-D form of~\eqref{eq:x} as, 
\begin{align}
x(u) = \sum_n^{N} x_n \phi_n^*(u), \text{ where, } x_n =  \frac{\langle s(u), \psi_n(u) \rangle}{\sigma_n}
\label{eq:hogmt}
\end{align}
Therefore, observing the similarity of \eqref{eq:dpc1} and~\eqref{eq:hogmt} we find that DPC is mathematically same as 2-D HOGMT precoding after scaling by the effective gain, $l(u)$.
\end{proof}
Figure~\ref{fig:eq} illustrates the equivalence shown in Theorem~\ref{thm:thm1} to provide an 
alternate implementation of DPC using HOGMT, as in~\eqref{eq:W}. Note that the non-linearity of DPC is due to the iterative feedback required by the 
Gram–Schmidt process as shown in Figure \ref{fig:eq}
Therefore, because of the equivalence and the linear implementation of 2-D HOGMT precoding by SVD provides significant computational advantage in practical implementation of DPC in MU-MIMO channels.

\noindent
\subsection{Beamformer optimization} 
From Theorem~\ref{thm:thm1}, the beamformer is obtained by designing an optimal pre-equalizer, $b(u)$ for the $u^\text{th}$ user. Then $x_n$ in~\eqref{eq:x_u} can be expressed as  
\begin{equation}
x_n =  \frac{\langle b(u), l(u), s(u), \psi_n(u) \rangle}{\sigma_n}
\end{equation} 
Specifically, if $b(u) = 1/{l(u)}$, then the DPC implemented using~\eqref{eq:x_u} is numerically equal to 2-D HOGMT in~\eqref{eq:2-Dhogmt}. Let $k(u) = b(u) l(u)$, which is the desired effective gain, and denote the discrete form of $k(u)$ as diagonal matrix $\mathbf{K}$, then~\eqref{eq:W} with beamformer is given by,

\begin{equation}
    \mathbf{W} = \mathbf{V} \mathbf{\Sigma}^{-1}\mathbf{U}^{H} \mathbf{K}
    \label{eq:W2}
\end{equation}
Then the beamformer design is to find a diagonal matrix $\mathbf{K}$ to replace $\mathbf{D_L}$ in~\eqref{eq:W}. The optimal $\mathbf{K}$ is obtained by the objective function $f(\cdot)$ under the power constraint $P$ as follows: 

\begin{equation}
\begin{aligned}
    \argmax_{\mathbf{K}} \quad & f(\mathbf{K})\\
    \text{s.t.} \quad & tr(\mathbf{W}\mathbf{W}^H) \leq P
\end{aligned}
\label{eq:opK}
\end{equation}

\section{Precoding order optimization for DPC}

\subsection{Optimum precoding order}

\begin{figure}
    \centering    \includegraphics[width=0.7\linewidth]{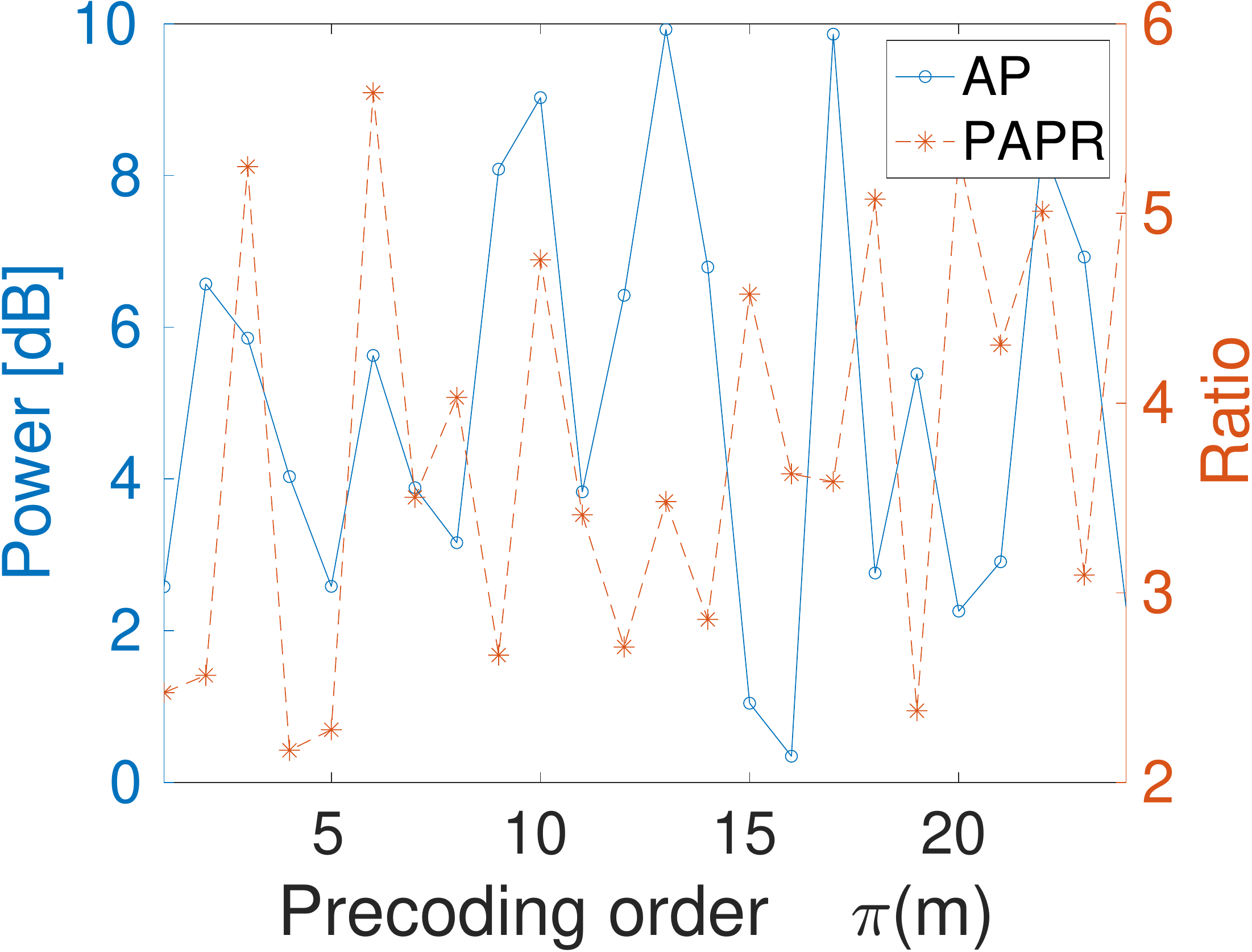}
    \caption{AP and PAPR for different precoding orders}
    \label{fig:Power}
\end{figure}

\begin{figure*}
\centering
\begin{subfigure}{.23\textwidth}
  \centering
\includegraphics[width=1\linewidth]{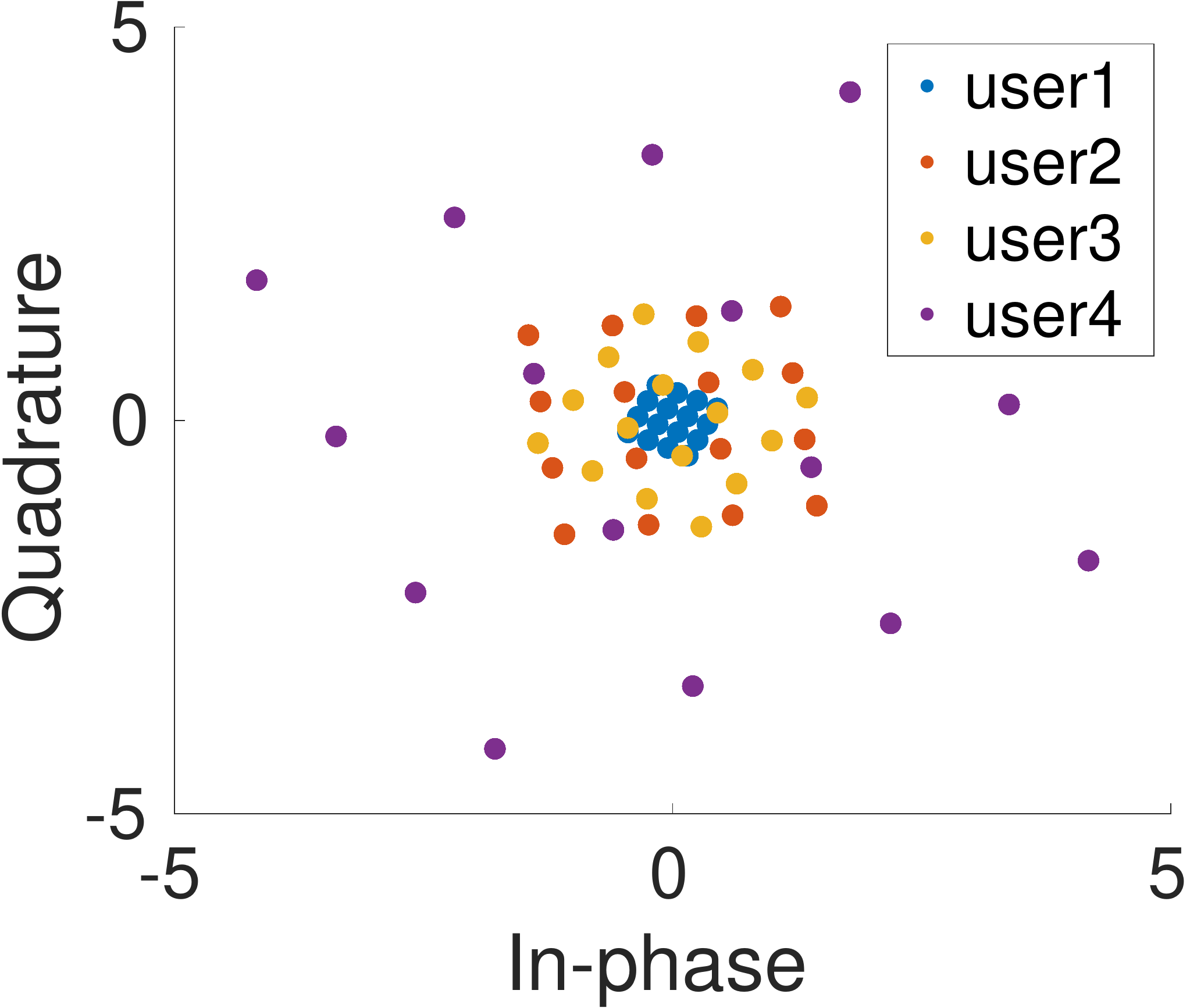}
  \caption{Maximal PAPR order, $m=6$}
  \label{fig:a}
\end{subfigure}
\quad
\begin{subfigure}{.23\textwidth}
  \centering
  \includegraphics[width=1\linewidth]{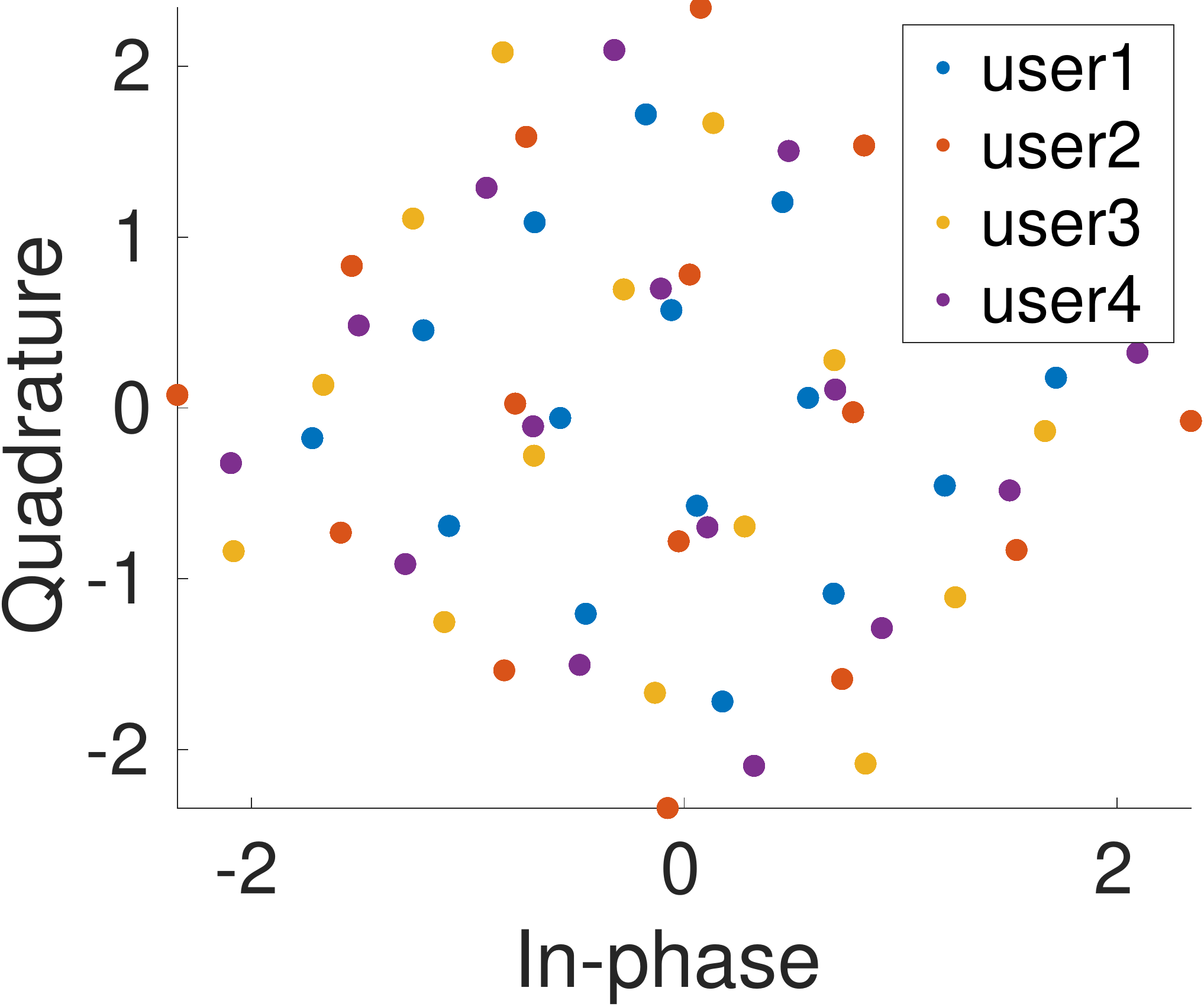}
  \caption{Minimal PAPR order, $m=4$}
  \label{fig:b}
\end{subfigure}
\quad
\begin{subfigure}{.23\textwidth}
  \centering
  \includegraphics[width=1\linewidth]{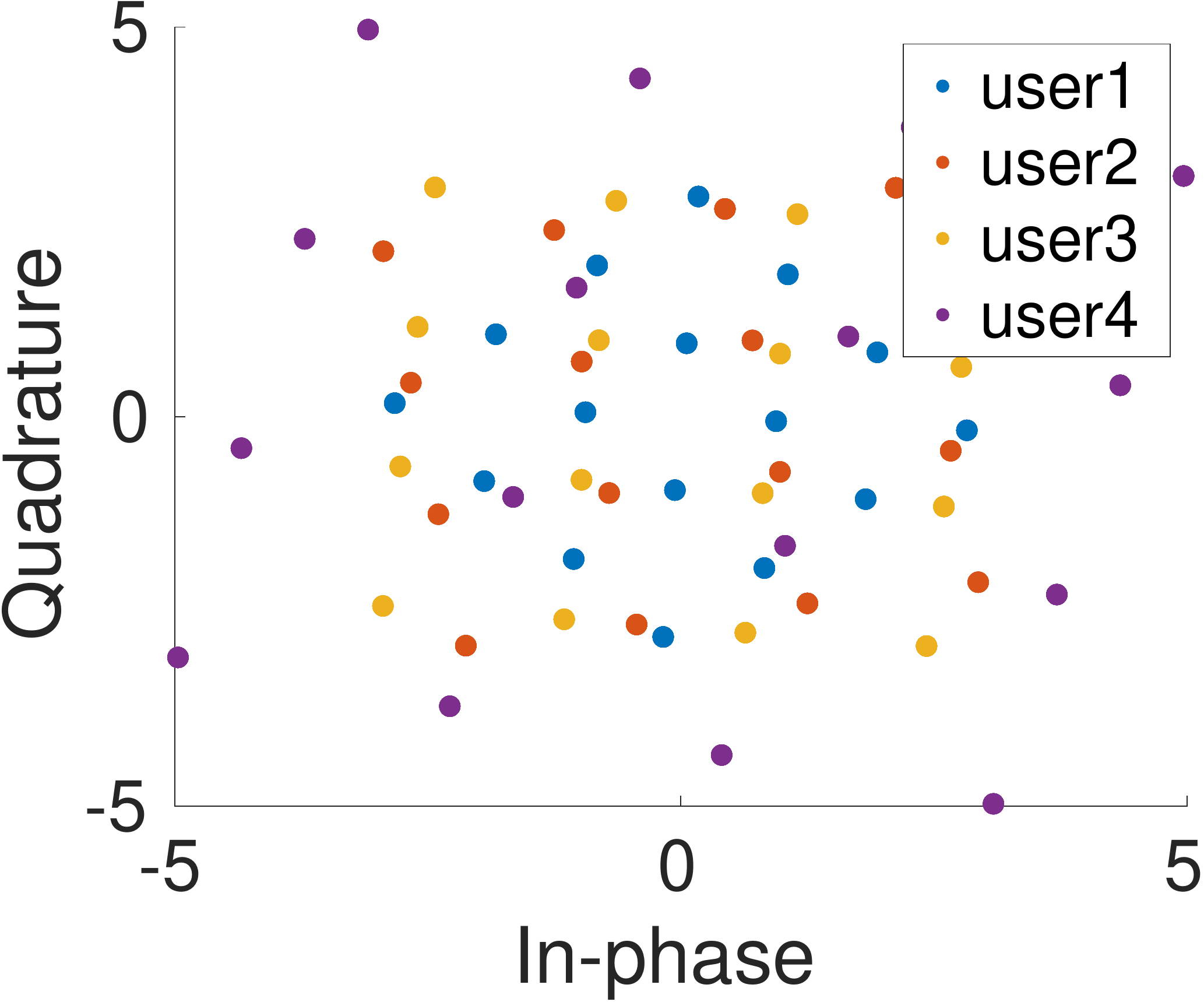}
  \caption{Maximal AP order, $m=13$
  } 
  \label{fig:c}
\end{subfigure}
\quad
\begin{subfigure}{.23\textwidth}
  \centering
  \includegraphics[width=1\linewidth]{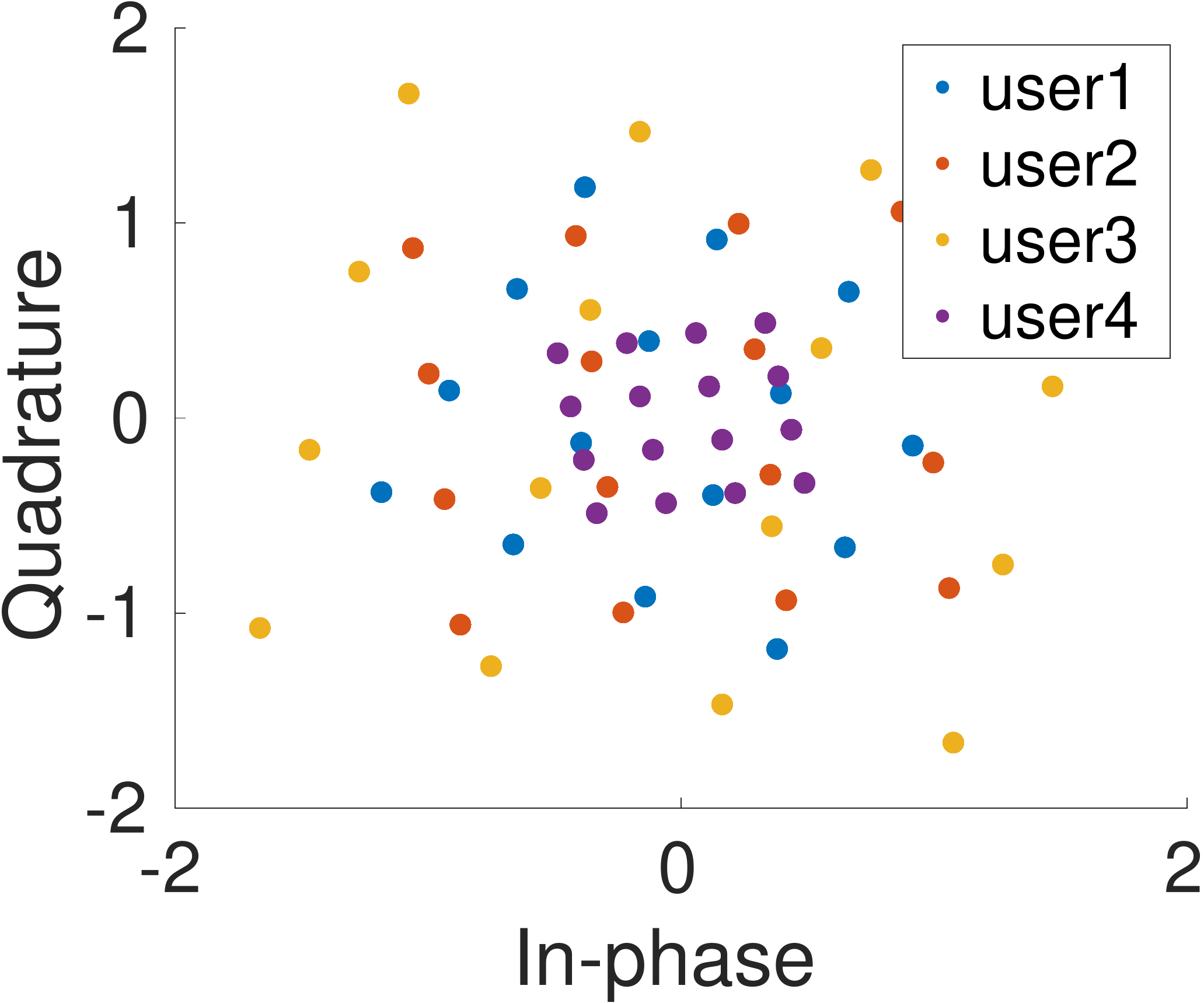}
  \caption{Minimal AP order, $m=16$}
  \label{fig:d}
\end{subfigure}
\caption{16-QAM constellation of 4-user DPC with different precoding order, $\pi(m)$ and power constraints}
\label{fig:4 orders}
\vspace{-10pt}
\end{figure*}

DPC treats each user as one layer and iteratively precodes on previous layers by treating the interference from previous layers as \textit{dirty}. This process is widely termed as \textit{writing on dirty paper}. Since each user channel is different, the order of these layers (users) affects the precoded signal~\cite{Michel2007Order}.

Figure~\ref{fig:Power} shows the Average Power (AP) and Peak-to-Average Power Ratio (PAPR) of DPC with 4 users for different precoding orders. The total number of precoding orders is $4! {=} 24$. We observe that AP and PAPR are varies with each precoding order. Figure~\ref{fig:4 orders} shows the constellation diagram for DPC with minima AP, maximal AP, minimal PAPR and maximal PAPR precoding orders, where the scatter plot of each layer is \textit{writing}     on previous layers. The gap between scatter plots of layers is largest in Figure~\ref{fig:a} and smallest in Figure~\ref{fig:b}, as the corresponding precoding order gives the maximal PAPR and minimal PAPR respectively. Figure~\ref{fig:c} has the maximal boundary for scatter plots while Figure~\ref{fig:d} has the minimal boundary, which suggests maximal AP and minimal AP respectively. 

The complexity of searching the optimal DPC order is known to be $O(N^3N!)$~\cite{Mao2020DPC}. The number of possible DPC precoding order is $N!$ for a $N$-user case. Now, for each order, DPC is repeated and then the precoded signal is compared based on the given constraint criteria to find the optimal order, which is extremely expensive computationally.

However, from Theorem~\ref{thm:thm1} and Figure \ref{fig:eq}, we find that the precoded matrix consists of the components from the SVD and LQ decomposition only. The relation of precoded matrix $\mathbf{W}$ and the permutation of channel matrix $\mathbf{H}$ is given by Lemma\ref{lemma:svd_per}.

\begin{lemma}
\label{lemma:svd_per}
Let $\mathbf{H}_{\pi(m)}$ be the permutation, by order ${\pi(m)}$ of a given channel matrix with SVD, $\mathbf{H}{=}\mathbf{U}\mathbf{\Sigma} \mathbf{V}^H$. Then the SVD of $\mathbf{H}_{\pi(m)}$ is given by permuting $\mathbf{U}$ by the same order ${\pi(m)}$, 
\begin{equation}
    \mathbf{H}_{\pi(m)} = \mathbf{U}_{\pi(m)} \mathbf{\Sigma} \mathbf{V}^H \label{eq:H_pi}
\end{equation}
\end{lemma}
\begin{proof}
Given a matrix $\mathbf{G}$ consisting of basis vectors $\{\mathbf{e}_n\}$,
\begin{equation}
    \mathbf{G} = \left[\begin{array}{c}
\mathbf{e}_1 \\
\mathbf{e}_2 \\
\vdots \\ 
\mathbf{e}_N
\end{array}\right] =\left[\begin{array}{llll}
1 & 0 & \ldots & 0\\
0 & 1 & \ldots & 0\\
\vdots & \vdots & \ddots & \vdots\\
0 & 0 & \ldots & 1
\end{array}\right]
\end{equation}
if $\{\mathbf{e}_n\}$ is rearranged by a given order $\pi(m)$, we get the permutation operator $\mathbf{G}_{\pi(m)}$. Then permutation of a matrix by  order $\pi(m)$ can seen as the multiplication by $\mathbf{G}_{\pi(m)}$. Thus the permuted matrix $\mathbf{H}_{\pi(m)}$ is given by

\begin{align}
    \mathbf{H}_{\pi(m)} & = \mathbf{G}_{\pi(m)} \mathbf{H} = \mathbf{G}_{\pi(m)}\mathbf{U} \mathbf{\Sigma} \mathbf{V}^{H}  \label{eq:SVD_1}
\end{align}

Note that $(\mathbf{G}_{\pi(m)}\mathbf{U})(\mathbf{G}_{\pi(m)}\mathbf{U})^H {=} \mathbf{G}_{\pi(m)} \mathbf{U} \mathbf{U}^H \mathbf{G}_{\pi(m)}^H = \mathbf{G}_{\pi(m)} \mathbf{I} \mathbf{G}_{\pi(m)}^H = \mathbf{I}$, implying $(\mathbf{G}_{\pi(m)}\mathbf{U})$ is an unitary matrix. Let $\mathbf{U}_{\pi(m)} = \mathbf{G}_{\pi(m)}\mathbf{U}$, then~\eqref{eq:SVD_1} can be rewritten as 
\begin{align}
    \mathbf{H}_{\pi(m)} & = \mathbf{U}_{\pi(m)} \mathbf{\Sigma} \mathbf{V}^{H} \label{eq:SVD_2} 
\end{align}
\end{proof}
\noindent
\subsection{Permutation on SVD and LQ decomposition} Lemma\ref{lemma:svd_per} shows that the SVD of the permuted matrix $\mathbf{H}_{\pi(m)}$ can be represented by the linear combination of permutation operator $\mathbf{G}_{\pi(m)}$ and SVD of $\mathbf{H}$, i.e.,
\begin{equation}
    \text{SVD}(\mathbf{H}_{\pi(m)}) = \mathbf{G}_{\pi(m)}\text{SVD}(\mathbf{H})
\end{equation}
This shows the linearity of SVD with respect to permutation. However, it is not the same for LQ decomposition, as the triangular matrix $\mathbf{L}$ is unable to maintain its triangular structure after permutation and since, $\mathbf{G}_{\pi(m)}\mathbf{L}$ is not a triangular matrix the LQ decomposition of a permuted matrix can not be obtained by permuting the decomposed component. Therefore,
\begin{equation}
    \text{LQ}(\mathbf{H}_{\pi(m)}) \neq \mathbf{G}_{\pi(m)}\text{LQ}(\mathbf{H})
    \label{eq:nl}
\end{equation}
\eqref{eq:nl} shows the non-linearity of LQ decomposition with respect to permutation, which requires conventional DPC to repeat LQ decomposition for each permutation of channel matrix $\mathbf{H}$ to find the optimal order for precoding   .  

From Lemma~\ref{lemma:svd_per}, we find that permutation of $\mathbf{H}$ 
only changes the order of the elements of 
$\mathbf{U}$. Therefore, using \eqref{eq:W2}, the precoding matrix under permutation, $\mathbf{W_m}$ is, 
\begin{equation}
    \mathbf{W}_m = \mathbf{V} \Sigma^{-1}\mathbf{U}_{\pi(m)}^{H} \mathbf{K} \label{eq:W1}
\end{equation}

Thus the optimal order can be obtained by an one-time DPC for an arbitrary order using 2-D HOGMT and then comparing each order by permuting the unitary matrix $\mathbf{U}$ and data signal $\mathbf{s}$. Thus we have the Theorem~\ref{thm:thm2}. 

\begin{theorem}
\label{thm:thm2}
The optimal precoding order for DPC with effective channel gains $\mathbf{K}$ is obtained by permuting the diagonal elements of $\mathbf{K}$.
\end{theorem}

\begin{proof}

The DPC precoding order is optimized as follows,
\begin{equation}
\begin{aligned}
     \argmin_{\{\pi(m)\}}\quad &
     g(\mathbf{x}_m)\\
    \text{s.t.} \quad & \mathbf{H}_{\pi(m)}\mathbf{W}_m = \mathbf{K}_{\pi(m)}
\end{aligned}
\label{eq:g}
\end{equation}
where $\mathbf{x}_m =  \mathbf{W}_m  \mathbf{s}_{\pi(m)}$, and $\mathbf{s}_{\pi(m)} = \mathbf{G}_{\pi(m)} \mathbf{s}$ is the permutation of data signal $\mathbf{s}$ by order $\pi(m)$. $g(\cdot)$ is the objective function according to the given criteria such as minimal AP, minimal PAPR, etc.

\noindent
\textit{Remark 1:}
(Diagonal permutation) Given a permutation operator $\mathbf{G}_{\pi(m)}$, which permutes the rows of matrix by the order ${\pi(m)}$, for a diagonal matrix $\mathbf{D}$, permuting the diagonal elements on the diagonal direction by the order ${\pi(m)}$ is, 
\begin{equation}
\begin{aligned}
    \mathbf{D}_{\pi(m)} = (\mathbf{G}_{\pi(m)} ^H \mathbf{D} \mathbf{G}_{\pi(m)}) 
\end{aligned}
\end{equation}
Then the precoded signal is given by,

\begin{align}
    \mathbf{x}_m &= \mathbf{W}_m \mathbf{S}_{\pi(m)} \nonumber \\ 
    & = \mathbf{V} \Sigma^{-1}(\mathbf{G}_{\pi(m)} \mathbf{U})^H \mathbf{K} \mathbf{G}_{\pi(m)} \mathbf{s} \nonumber \\ 
    & =  \mathbf{V} \mathbf{\Sigma}^{-1}\mathbf{U}^H \underbrace{(\mathbf{G}_{\pi(m)} ^H \mathbf{K} \mathbf{G}_{\pi(m)})}_{\text{Diagonal permutation}} \mathbf{s} \nonumber \\
    & = \mathbf{V}\mathbf{\Sigma}^{-1}\mathbf{U}^H \mathbf{K}_{\pi(m)} \mathbf{s} \label{eq:Xtild}
\end{align}
where $\mathbf{K}_{\pi(m)}$ is a diagonal matrix having the same elements of $\mathbf{K}$ with diagonal entries ordered by $\pi(m)$. 
\end{proof}

Theorem~\ref{thm:thm2} shows that the solution of optimal precoding order with respect to arbitrary objective function $g(\cdot)$ can be obtained by looping over all precoding orders, where for each precoding order, the proposed method can avoid repeating the decomposition by simply permuting a diagonal matrix.

\subsection{Convergence of beamforming and precoding orders for the same strategy}
Specifically, if the criteria of the beamforming is the optimal power allocation, the beamforming solution already achieves minimal power precoding order as shown in Corollary~\ref{col:col1}. 
\begin{corollary}
\label{col:col1}
The optimal power allocation strategy achieves minimal power precoding order.
\end{corollary}
\begin{proof}
In Theorem~\ref{thm:thm1}, given power constraint $\text{tr}(\mathbf{W}\mathbf{W}^H) \leq P$, the precoding matrix $\mathbf{W}$ with optimal power allocation beamforming is obtained by
\begin{equation}
    \mathbf{W} = \mathbf{V} \mathbf{\Sigma}^{-1}\mathbf{U}^{H} \mathbf{K}
\end{equation}
where, $\mathbf{K}=\text{diag}(k_1,\ldots,k_N)$, and $k_n$ is the effective gain for user $n$ and can be designed by water-filling algorithm as

\begin{align}
    k_n &= \sqrt{p_n \lambda_n}, \quad \text{where,} \quad
    p_n = \left(\mu - \frac{1}{\lambda_n}\right)^+ \label{eq:kandp}
\end{align}
where $\mu$ is a constant to ensure power constraint, and $\lambda_n = \sigma_n^2$ is $n^\text{th}$ eigenvalue where $\sigma_n$ is the $n^\text{th}$ diagonal element of $\Sigma$. $(x)^+$ is defined as $\text{max}(x,0)$. 

On the other hand, to find the optimal precoding order with respect to minimal power, set $g(\mathbf{x}) {=} \mathbb{E}\{|\mathbf{x}|^2\}$, then \eqref{eq:g}, 

\begin{equation}
\begin{aligned}
     \argmin_{\{\pi(m)\}}\quad &
     \mathbb{E}\{|\mathbf{x}_m|^2\} \\
    \text{s.t.} \quad & \mathbf{H}_{\pi(m)}\mathbf{W}_m = \mathbf{K}_{\pi_m}
\end{aligned}
\label{eq:Ex}
\end{equation}
Substituting \eqref{eq:Xtild} in \eqref{eq:Ex}, we have 
\begin{equation}
\begin{aligned}
    \argmin_{\{\pi(m)\}}\quad & \sum_n^N \frac{ k_{n,\pi(m)}^2}{\lambda_n} \label{eq:korder}
\end{aligned}
\end{equation}
where, $k_{n,\pi(m)}$ is $n^\text{th}$ diagonal element of $\mathbf{K}_{\pi(m)}$. Thus, the optimal order in \eqref{eq:Ex} is obtained by simply permuting $\mathbf{K}$. 

\eqref{eq:korder} suggests that the order $\{\pi(m)\}$ ensures
$\{k_{n,\pi(m)}\}$ has the same magnitude order as $\{\lambda_n\}$, achieves the optimal solution. 
Meanwhile, as $k_n$ in the original $\mathbf{K}$ given by~\eqref{eq:kandp} has the positive relation with $\lambda_n$, it is ranked by the same order as $\lambda_n$.  which is the solution of the optimization \eqref{eq:korder}. Thus the original order of $k_n$ is already optimal.
\end{proof}

The beamforming optimization with optimal power allocation strategy target to maximize the energy efficiency. The precoding order optimization with the minimal power criteria target to minimize the signal power with the effective gain unchanged as in~\eqref{eq:Ex}, which also maximize the energy efficiency. They are reasonable to converge to the same solution. Thus the equivalence shown in Corollary 1 validate the correctness of the proposed precoding order optimization.

\subsection{Capacity achieving by diagonal permutation}
The procedure to implement the equivalent DPC (capacity-achieving technique) with precoding order optimization is given in Algorithm~1. First, decompose the CSI by SVD as $\mathbf{H} {=} \mathbf{U}\mathbf{\Sigma}\mathbf{V}^H$ using  Theorem~\ref{thm:thm1} and then constitute $\mathbf{W}$ by the decomposed components with the channel effective gain $\mathbf{K}$ as in~\eqref{eq:W2}, where $\mathbf{K}$ is obtained by solving~\eqref{eq:opK}. To find the optimal precoding order, we collect the all permutation orders $\{\pi(m)\}_{m=1}^M$, where for $N$ users, there are $M {=} N!$ precoding orders. For the precoding order $\pi(m)$, the corresponding precoded signal $\mathbf{x}_m {=} \mathbf{W}_m \mathbf{s}_{\pi(m)}$ is obtained from~\eqref{eq:Xtild}, which only diagonally permutes $\mathbf{K}$ by order $\pi(m)$ as $\mathbf{K_{\pi(m)} = \mathbf{G}_{\pi(m)}^H \mathbf{K} \mathbf{G}_{\pi(m)}}$. The precoded signal with optimal precoding order is obtained by looping over all orders and comparing the corresponding precoded signal by the decision function $g(\cdot)$. In practice, $g(\cdot)$ is based on desired criteria such as minimal AP or minimal PAPR.

\begin{algorithm}
\caption{Capacity achieving by diagonal permutation}
\label{alg:1}
\textbf{Input}: Data $\mathbf{s}$, CSI $\mathbf{H}$, beamforming optimization function $f(\cdot)$, power limitation $P$, decision function $g(\cdot)$\;

\textbf{Output}: Precoded signal $\mathbf{x}$\;
Decompose $\mathbf{H}$ by SVD decomposition $[\mathbf{U}, \mathbf{\Sigma}, \mathbf{V}] = \text{SVD}(\mathbf{H})$\;
Solve $\argmax_{\mathbf{K}} f(\mathbf{K})$, under the constraint $\text{tr}(\mathbf{W}\mathbf{W}^H) \leq P$, where $\mathbf{W} = \mathbf{V} \mathbf{\Sigma}^{-1}\mathbf{U}^{H} \mathbf{K}$ to get $\mathbf{K}$\;
Collect diagonal permutation orders $\{\pi(\cdot)\}$ of $\mathbf{K}$\;
Initialize the order index $m = 1$\;
Permute $\mathbf{K}$ by order $\pi(m)$ to get $\mathbf{K}_{\pi(m)}$\; 
Initialize $\mathbf{x} = \mathbf{V}\mathbf{\Sigma}^{-1}\mathbf{U}^H \mathbf{K}_{\pi(m)}$\;
\While{$m < \text{size}(\{\pi(\cdot)\})$}{ 
  $m = m + 1$ \;
  \text{Update} $\mathbf{K}_{\pi(m)}$ \text{by permuting} $\mathbf{K}$ \text{by order} $\pi(m)$\;
  $\mathbf{x}_m = \mathbf{V}\mathbf{\Sigma}^{-1}\mathbf{U}^H \mathbf{K}_{\pi(m)}$\;
  \text{Update} $\mathbf{x}$ \text{by decision function} $\mathbf{x} = g(\mathbf{x}_m, \mathbf{x})$\; 
}
\end{algorithm}

\begin{figure*}[t]
\centering
\begin{subfigure}{.3\textwidth}
\centering
      \includegraphics[width=0.8\linewidth]{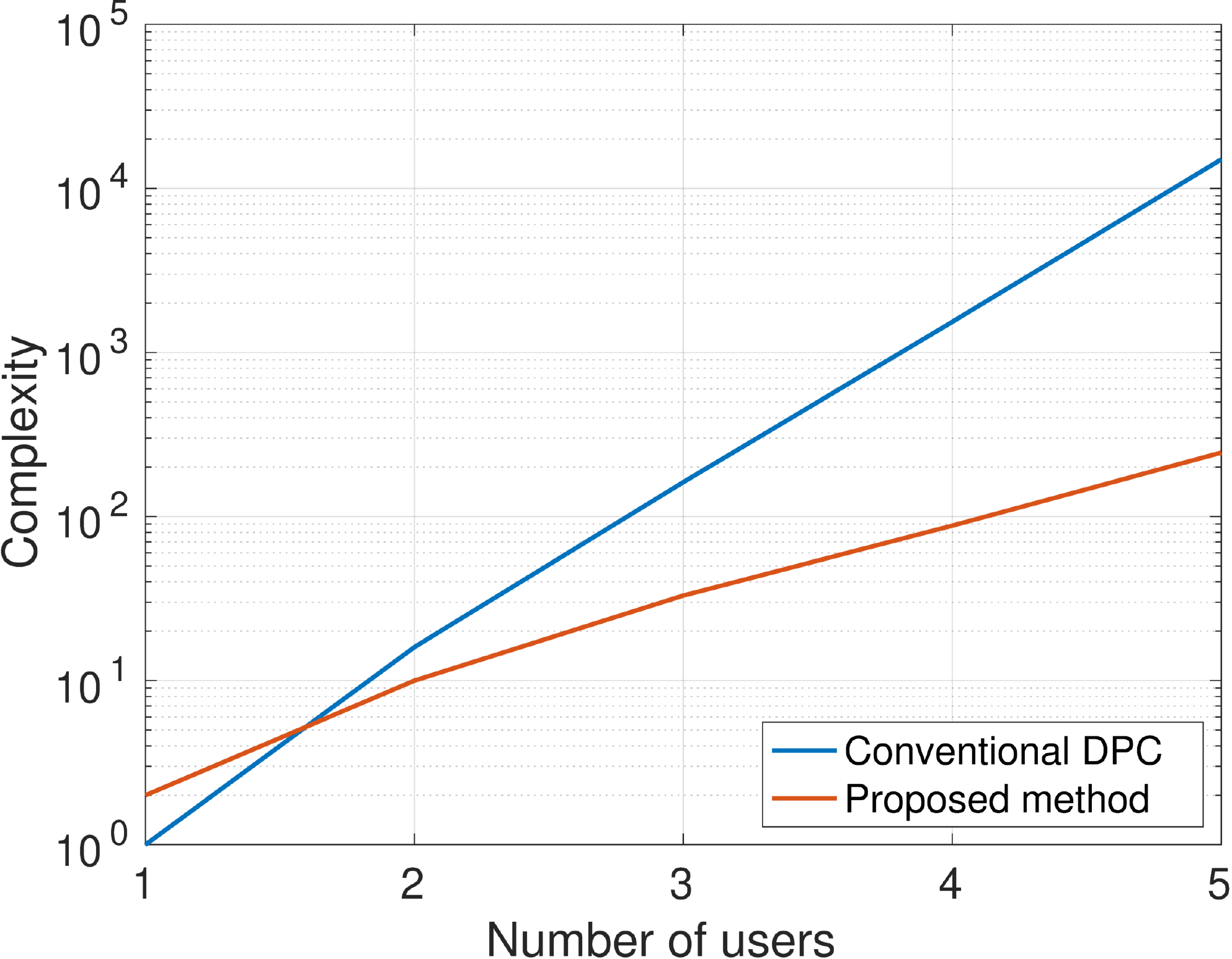}
    \caption{Complexity ($\mathcal{O}$) of conventional DPC and Algorithm~\ref{alg:1} with respect to $N$}
    \label{fig:complexity}
\end{subfigure}
\quad
\begin{subfigure}{.3\textwidth}
     \centering
    \includegraphics[width=0.8\linewidth]{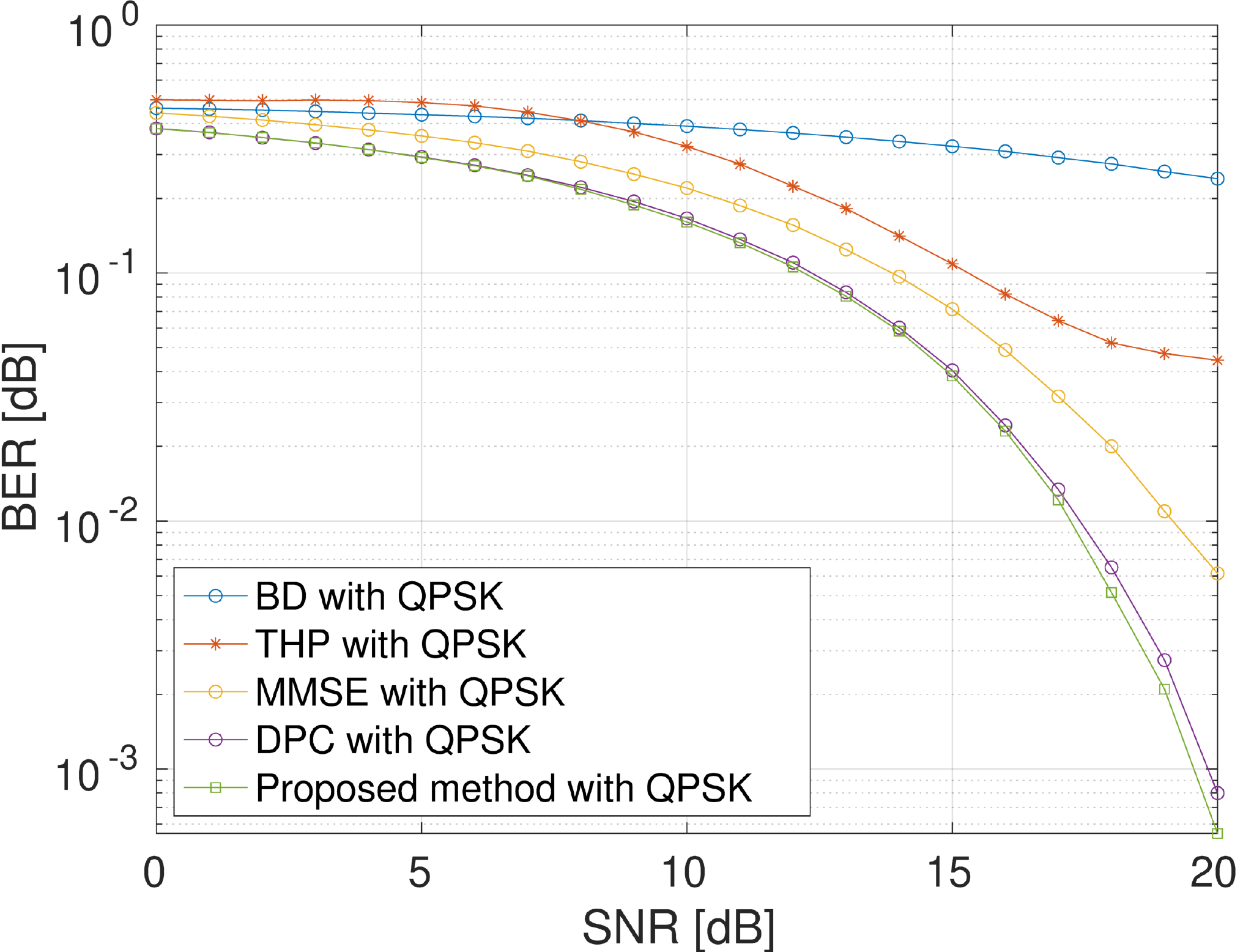}
    \caption{BER of BD, THP, MMSE, conventional DPC and Algorithm~\ref{alg:1} for QPSK}
    \label{fig:BER}
\end{subfigure}
\quad
\begin{subfigure}{.3\textwidth}
     \centering
    \includegraphics[width=0.8\linewidth]{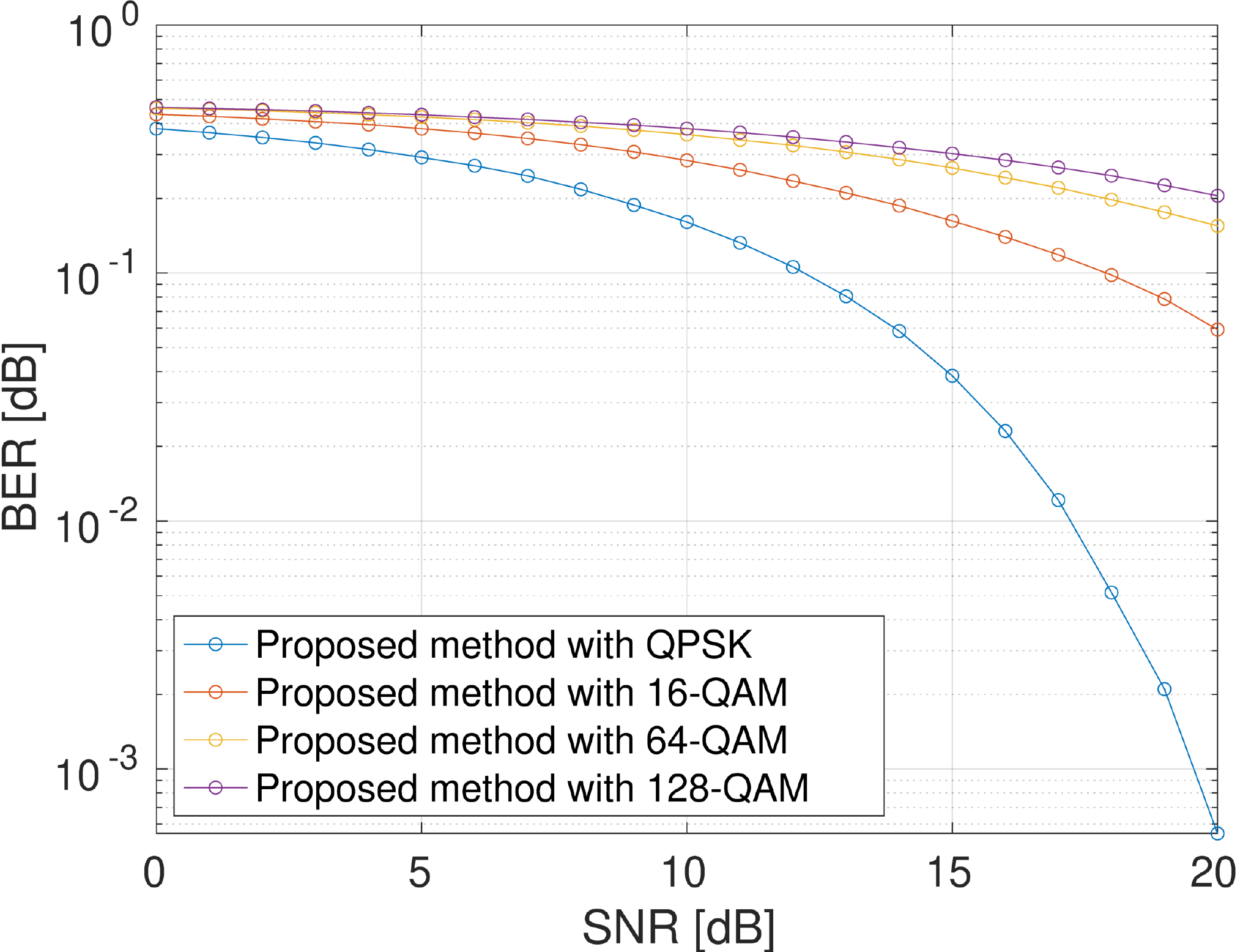}
    \caption{Algorithm~\ref{alg:1} for QPSK, 16-QAM, 64-QAM and 128-QAM schemes}
    \label{fig:ber_multi}
\end{subfigure}
\caption{Comparison of complexity and BER}
\vspace{-10pt}
\end{figure*}
Conventional DPC implementation needs to repeat DPC for each precoding order, whose complexity is $\mathcal{O}(N^3N!)$ for $N$ users case. In contrast the proposed algorithm requires a one-time DPC and then search for the optimal precoding order by just permuting the channel gain matrix that has a computational complexity of $\mathcal{O}(N^3+N!)$. Figure~\ref{fig:complexity} shows that for multi-user cases, i.e., $N \geq 2$, the proposed methods always achieve lower complexity than conventional DPC. For users number $N=5$, the complexity ratio of conventional DPC over proposed method is  ${\approx}20$ dB, which is very encouraging.

\section{Results}

We validate the equivalence of the proposed methods and conventional DPC with perfect CSI at the transmitter using MATLAB simulations. The number of the transmitter antennas and the users are both $N {=} 10$, where each user is equipped with one antenna. The coefficient of the channel matrix is generated by standard Gaussian distribution. The effective channel gain of both methods are normalized. Figure~\ref{fig:BER} compares the BER of Block Diagonalization (BD), Tomlinson-Harashima Precoding (THP), Minimum Mean Square Error (MMSE) precoding, conventional DPC and proposed method for QPSK. It is evident that the proposed method achieves the similar result as conventional DPC and outperforms other techniques, supporting the theoretical equivalence discussed earlier. The proposed method for various modulation schemes is compared in Figure~\ref{fig:ber_multi}.


\section{Conclusion}
\label{sec:conclusion}
In this paper, we show the equivalence of DPC and 2-D HOGMT precoding for MU-MIMO channels and give an alternate low-complexity implementation based on SVD decomposition to replace the iterative method based on Gram-Schmidt processes and LQ decomposition. Then we show the difference between SVD decomposition and LQ decomposition with respect to permutation, where an unitary matrix after permutation is still unitary but a triangular matrix cannot maintain its structure under the same permutation. This difference suggests that conventional method needs to repeat DPC for different precoding orders while the proposed method just needs one-time DPC and then searches the optimal precoding order by permuting a component of the precoding matrix, which is also shown the convergence. Simulations show that our method is able to achieve the same BER performance as DPC but with less complexity. For $N=5$ users case, the proposed method achieves near $20$ dB lower complexity.






\bibliographystyle{IEEEtran}
\bibliography{references}

\begin{thebibliography}{10}
\providecommand{\url}[1]{#1}
\csname url@samestyle\endcsname
\providecommand{\newblock}{\relax}
\providecommand{\bibinfo}[2]{#2}
\providecommand{\BIBentrySTDinterwordspacing}{\spaceskip=0pt\relax}
\providecommand{\BIBentryALTinterwordstretchfactor}{4}
\providecommand{\BIBentryALTinterwordspacing}{\spaceskip=\fontdimen2\font plus
\BIBentryALTinterwordstretchfactor\fontdimen3\font minus
  \fontdimen4\font\relax}
\providecommand{\BIBforeignlanguage}[2]{{%
\expandafter\ifx\csname l@#1\endcsname\relax
\typeout{** WARNING: IEEEtran.bst: No hyphenation pattern has been}%
\typeout{** loaded for the language `#1'. Using the pattern for}%
\typeout{** the default language instead.}%
\else
\language=\csname l@#1\endcsname
\fi
#2}}
\providecommand{\BIBdecl}{\relax}
\BIBdecl

\bibitem{Zou2022ICC}
Z.~Zou, M.~Careem, A.~Dutta, and N.~Thawdar, ``Unified characterization and
  precoding for non-stationary channels,'' in \emph{ICC 2022 - IEEE
  International Conference on Communications}, 2022, pp. 5140--5146.

\bibitem{ZouTCOM23}
{Z. Zou}, M.~Careem, A.~Dutta, and N.~Thawdar, ``{Joint Spatio-Temporal
  Precoding for Practical Non-Stationary Wireless Channels},'' \emph{IEEE
  Transactions on Communications}, pp. 1--1, 2023.

\bibitem{Vu2007Precoding}
M.~Vu and A.~Paulraj, ``{MIMO Wireless Linear Precoding},'' \emph{IEEE Signal
  Processing Magazine}, vol.~24, no.~5, pp. 86--105, 2007.

\bibitem{CostaDPC1983}
M.~Costa, ``Writing on dirty paper (corresp.),'' \emph{IEEE Transactions on
  Information Theory}, vol.~29, no.~3, pp. 439--441, 1983.

\bibitem{Cho2010MIMObook}
Y.~S. Cho, J.~Kim, W.~Y. Yang, and C.~G. Kang, \emph{{MIMO-OFDM Wireless
  Communications with MATLAB}}.\hskip 1em plus 0.5em minus 0.4em\relax Wiley
  Publishing, 2010.

\bibitem{Ma2016MUMIMO}
Y.~Ma, A.~Yamani, N.~Yi, and R.~Tafazolli, ``Low-complexity mu-mimo nonlinear
  precoding using degree-2 sparse vector perturbation,'' \emph{IEEE Journal on
  Selected Areas in Communications}, vol.~34, no.~3, pp. 497--509, 2016.

\bibitem{Caire2003capacity}
G.~Caire and S.~Shamai, ``{On the achievable throughput of a multiantenna
  Gaussian broadcast channel},'' \emph{IEEE Transactions on Information
  Theory}, vol.~49, no.~7, pp. 1691--1706, 2003.

\bibitem{Vishwanath2003Dual}
S.~Vishwanath, N.~Jindal, and A.~Goldsmith, ``Duality, achievable rates, and
  sum-rate capacity of gaussian mimo broadcast channels,'' \emph{IEEE
  Transactions on Information Theory}, vol.~49, no.~10, pp. 2658--2668, 2003.

\bibitem{LeeDPC2007}
J.~Lee and N.~Jindal, ``{High SNR Analysis for MIMO Broadcast Channels: Dirty
  Paper Coding Versus Linear Precoding},'' \emph{IEEE Transactions on
  Information Theory}, vol.~53, no.~12, pp. 4787--4792, 2007.

\bibitem{Weingarten2006Capacity}
H.~Weingarten, Y.~Steinberg, and S.~Shamai, ``The capacity region of the
  gaussian multiple-input multiple-output broadcast channel,'' \emph{IEEE
  Transactions on Information Theory}, vol.~52, no.~9, pp. 3936--3964, 2006.

\bibitem{Rashid1998Bm}
F.~Rashid-Farrokhi, K.~R. Liu, and L.~Tassiulas, ``{Transmit beamforming and
  power control for cellular wireless systems},'' \emph{IEEE Journal on
  Selected Areas in Communications}, vol.~16, no.~8, pp. 1437--1450, 1998.

\bibitem{Schubert2004Bm}
M.~Schubert and H.~Boche, ``Solution of the multiuser downlink beamforming
  problem with individual sinr constraints,'' \emph{IEEE Transactions on
  Vehicular Technology}, vol.~53, no.~1, pp. 18--28, 2004.

\bibitem{Schubert2005Bm}
{M. Schubert} and {H. Boche}, ``{Iterative multiuser uplink and downlink
  beamforming under SINR constraints},'' \emph{IEEE Transactions on Signal
  Processing}, vol.~53, no.~7, pp. 2324--2334, 2005.

\bibitem{Michel2007Order}
\BIBentryALTinterwordspacing
T.~Michel and G.~Wunder, ``Transmitter and precoding order optimization for
  nonlinear downlink beamforming,'' 2007. [Online]. Available:
  \url{https://arxiv.org/abs/cs/0702122}
\BIBentrySTDinterwordspacing

\bibitem{Almers2007ChannelSurvey}
P.~Almers, E.~Bonek, A.~Burr, N.~Czink, m.~Debbah, V.~Degli-Esposti,
  H.~Hofstetter, P.~Kyösti, D.~Laurenson, G.~Matz, A.~Molisch, C.~Oestges, and
  H.~Özcelik, ``{Survey of Channel and Radio Propagation Models for Wireless
  MIMO Systems.}'' \emph{EURASIP J. Wireless Comm. and Networking}, vol. 2007,
  01 2007.

\bibitem{Matz2005NS}
{G. Matz}, ``On non-wssus wireless fading channels,'' \emph{IEEE Transactions
  on Wireless Communications}, vol.~4, no.~5, pp. 2465--2478, 2005.

\bibitem{MATZ20111}
F.~Hlawatsch and G.~Matz, \emph{Wireless Communications Over Rapidly
  Time-Varying Channels}, 1st~ed.\hskip 1em plus 0.5em minus 0.4em\relax USA:
  Academic Press, Inc., 2011.

\bibitem{Mao2020DPC}
Y.~Mao and B.~Clerckx, ``{Beyond Dirty Paper Coding for Multi-Antenna Broadcast
  Channel With Partial CSIT: A Rate-Splitting Approach},'' \emph{IEEE
  Transactions on Communications}, vol.~68, no.~11, pp. 6775--6791, 2020.

\end{thebibliography}




\end{document}


\def\eg{\mbox{\em e.g.}, }

\title{Proofs and Supplementary Material:\\ Unified Characterization and Precoding for Non-Stationary Channels}

\author{
\begin{tabular}[t]{c@{\extracolsep{8em}}c} 
Zhibin Zou, Maqsood Careem, Aveek Dutta & Ngwe Thawdar \\
Department of Electrical and Computer Engineering & US Air Force Research Laboratory \\ 
University at Albany SUNY, Albany, NY 12222 USA & Rome, NY, USA \\
\{{zzou2, mabdulcareem, adutta\}@albany.edu} & 
ngwe.thawdar@us.af.mil
\end{tabular}
\vspace{-3ex}
}


    
\maketitle









\setcounter{equation}{34}
\appendices
\label{appendix:precoding}





\noindent
\textbf{Instructions:}
Equations (1)--(34) refer to the equations from the main manuscript (``Unified Characterization and Precoding for Non-Stationary Channels" accepted for publication at IEEE ICC 2022). This document provides the supplementary material including a comprehensive related work, the complete proofs and extended evaluation results to support the main manuscript.  


\section{Related work}
\label{App:related_NLP}


We categorize the related work into three categories:

\noindent
\textbf{Characterization of Non-Stationary Channels:}
Wireless channel characterization in the literature typically require several local and global (in space-time dimensions) higher order statistics to characterize or model non-stationary channels, due to their time-varying statistics. 
These statistics cannot completely characterize the non-stationary channel, however are useful in reporting certain properties that are required for the application of interest such as channel modeling, assessing the degree of stationarity etc.
Contrarily, we leverage the 2-dimensional eigenfunctions that are decomposed from the most generic representation of any wireless channel as a spatio-temporal channel kernel.
These spatio-temporal eigenfunctions can be used to extract any higher order statistics of the channel as demonstrated in Section \rom{3}, and hence serves as a complete characterization of the channel.
Furthermore, since this characterization can also generalize to stationary channels, it is a unified characterization for any wireless channel.
Beyond characterizing the channel, these eigenfunctions are the core of the precoding algorithm.

\noindent
\textbf{Precoding Non-Stationary Channels:}
Although precoding non-stationary channels is unprecedented in the literature \cite{AliNS0219}, we list the most related literature for completeness. 
The challenge in precoding non-stationary channels is the lack of accurate models of the channel and the (occasional) CSI feedback does not fully characterize the non-stationarities in its statistics. This leads to suboptimal performance using state-of-the-art precoding techniques like Dirty Paper Coding which assume that complete and accurate knowledge of the channel is available, while the CSI is often outdated in non-stationary channels. 
While recent literature present  attempt to deal with imperfect CSI by modeling the error in the CSI \cite{HatakawaNLP2012, HasegawaTHP2018, GuoTHP2020, DietrichTHP2007, CastanheiraPGS2013, WangTHP2012, MazroueiTDVP2016, Jacobsson1DAC2017}, they are limited by the assumption the channel or error statistics are stationary or WSSUS at best.
Another class of literature, attempt to deal with the impact of outdated CSI~\cite{AndersonLP2008,Zeng2012LP} in time-varying channels by quantifying this loss or relying statistical CSI. These methods are not directly suitable for non-stationary channels, as the time dependence of the statistics may render the CSI (or its statistics) stale, consequently resulting in precoding error. 

\noindent
\textbf{Space-Temporal Precoding:}
While, precoding has garnered significant research, spatio-temporal interference is typically treated as two separate problems, where spatial precoding at the transmitter aims to cancel inter-user and inter-antenna interference, while equalization at the receiver mitigates inter-carrier and inter-symbol interference.
Alternately, \cite{hadani2018OTFS} proposes to modulate the symbols such that it reduces the cross-symbol interference in the delay-Doppler domain, but requires equalization at the receiver to completely cancel such interference in practical systems. 
Moreover, this approach cannot completely minimize the joint spatio-temporal interference that occurs in non-stationary channels since their statistics depend on the time-frequency domain in addition to the delay-Dopper domain (explained in Section \rom{2}).
While spatio-temporal block coding techniques are studied in the literature \cite{Cho2010MIMObook} they add redundancy and hence incur a communication overhead to mitigate interference, which we avoid by precoding. 
These techniques are capable of independently canceling the interference in each domain, however are incapable of mitigating interference that occurs in the joint spatio-temporal domain in non-stationary channels. 
We design a joint spatio-temporal precoding that leverages the extracted 2-D eigenfunctions from non-stationary channels to mitigate interference that occurs on the joint space-time dimensions, which to the best of our knowledge is unprecedented in the literature.

\section{Proofs on Unified Characterization}
\label{app:characterization}

\subsection{Proof of Lemma 1: Generalized Mercer's Theorem}
\label{App:gmt}

\begin{proof}
Consider a 2-dimensional process $K(t,t') \in L^2(Y \times X)$, where $Y(t)$ and $X(t')$ are square-integrable zero-mean random processes with covariance function $K_{Y}$ and $K_{X}$, respecly. 
The projection of $K(t, t')$ onto $X(t')$ is obtained as in \eqref{eq:projection},
\begin{align}
\label{eq:projection}
    & C(t) = \int K(t, t') X(t') ~dt'
\end{align}

Using \textit{Karhunen–Loève Transform} (KLT), $X(t')$ and $C(t)$ are both decomposed as in \eqref{eq:X_t} and \eqref{eq:C_t}, 
\begin{align}
    &X(t') = \sum_{i = 1}^{\infty} x_{i} \phi_{i}(t') \label{eq:X_t}\\
    &C(t) = \sum_{j = 1}^{\infty} c_{j} \psi_{j}(t) \label{eq:C_t}
\end{align}
where $x_i$ and $c_j$ are both random variables with $\mathbb{E}\{x_i x_{i'}\} {=} \lambda_{x_i} \delta_{ii'}$ and $\mathbb{E}\{c_j c_{j'}\} {=} \lambda_{c_j} \delta_{jj'}$.  $\{\lambda_{x_i}\}$, $\{\lambda_{x_j}\}$ $\{\phi_i(t')\}$ and $\{\psi_j(t)\}$ are eigenvalues and eigenfuncions, respectively.
%

Let us denote $n{=}i{=}j$ and $\sigma_n {=} \frac{c_n}{x_n}$, and assume that $K(t,t')$ can be expressed as in \eqref{eq:thm_K_t},
\begin{align}
\label{eq:thm_K_t}
    K(t,t') = \sum_n^\infty \sigma_n \psi_{n}(t) \phi_{n}(t')
\end{align}
We show that \eqref{eq:thm_K_t} is a correct representation of $K(t,t')$ by proving \eqref{eq:projection} holds under this definition. 
We observe that by substituting \eqref{eq:X_t} and \eqref{eq:thm_K_t} into the right hand side of \eqref{eq:projection} we have that,
\begin{align}
    & \int K(t, t') X(t') ~dt' \nonumber \\
    & = \int \sum_n^\infty \sigma_n \psi_{n}(t) \phi_{n}(t') \sum_{n}^{\infty} x_{n} \phi_{n}(t') ~dt' \nonumber \\
    & = \int \sum_n^\infty \sigma_n x_n \psi_n(t) |\phi_n(t')|^2 \nonumber\\
    & + \sum_{n'\neq n}^ \infty \sigma_{n} x_{n'} \psi_{n}(t) \phi_{n}(t') \phi_{n'}^*(t') ~d t' \nonumber \\
    & = \sum_n^\infty c_n \psi_n(t) = C(t)
\end{align}
which is equal to the left hand side of \eqref{eq:projection}. 
Therefore, \eqref{eq:thm_K_t} is a correct representation of $K(t,t')$.

\end{proof}

\subsection{Proof of Theorem 1: High Order Generalized Mercer's Theorem (HOGMT}
\label{app:hogmt}

\begin{proof}
Given a 2-D process $X(\gamma_1, \gamma_2)$, the eigen-decomposition using Lemma 1 is given by,
\begin{equation}
\label{eq:thm1_1}
    X(\gamma_1, \gamma_2) = \sum_{n}^{\infty} x_{n} e_n(\gamma_1) s_n(\gamma_2)
\end{equation}

Letting $\psi_n(\gamma_1,\gamma_2) {=} e_n(\gamma_1) s_n(\gamma_2)$, and substituting it in \eqref{eq:thm1_1} we have that,

\begin{equation}
\label{eq:2d_klt}
    X(\gamma_1, \gamma_2) = \sum_{n}^{\infty} x_{n} \phi_n(\gamma_1,\gamma_2)
\end{equation}
where $\phi_n(\gamma_1,\gamma_2)$ are 2-D eigenfunctions with the property \eqref{eq:prop1}.
\begin{equation}
\label{eq:prop1}
\iint \phi_n(\gamma_1,\gamma_2) \phi_{n'}(\gamma_1,\gamma_2) ~d\gamma_1 ~d\gamma_2 = \delta_{nn'} 
\end{equation}

We observe that \eqref{eq:2d_klt} is the 2-D form of KLT. With iterations of the above steps, we obtain \textit{Higher-Order KLT} for $X(\gamma_1,\cdots,\gamma_Q)$ and $C(\zeta_1,\cdots,\zeta_P)$ as given by,
\begin{align}
   & X(\gamma_1,\cdots,\gamma_Q) = \sum_{n}^{\infty} x_{n} \phi_n(\gamma_1,\cdots,\gamma_Q) \\
   & C(\zeta_1,\cdots,\zeta_P) = \sum_{n}^{\infty} c_{n} \psi_n(\zeta_1,\cdots,\zeta_P)
\end{align}
where $C(\zeta_1,\cdots,\zeta_P)$ is the projection of $X(\gamma_1,\cdots,\gamma_Q)$ onto $K(\zeta_1,\cdots,\zeta_P; \gamma_1,\cdots, \gamma_Q)$.

Then following similar steps as in Appendix~\ref{App:gmt} we get \eqref{eq:col}. 
\begin{align}
\label{eq:col}
& K(\zeta_1,\cdots,\zeta_P; \gamma_1,\cdots, \gamma_Q) \nonumber \\
& = \sum_{n}^ \infty \sigma_n \psi_n(\zeta_1,\cdots,\zeta_P) \phi_n(\gamma_1,\cdots, \gamma_Q)
\end{align}
\end{proof}

\section{Proofs on Eigenfunction based Precoding}
\label{app:precoding}

\subsection{Proof of Lemma 2}
\label{app:lem1}

\begin{proof}
Using 2-D KLT as in (13), $x(u,t)$ is expressed as,
\begin{equation}
    x(u,t) = \sum_{n}^ \infty x_n \phi_n(u,t)
\end{equation}
where $x_n$ is a random variable with $E\{x_n x_{n'}\}{=} \lambda_n \sigma_{nn'} $ and $\phi_n(u,t)$ is a 2-D eigenfunction. 

Then the projection of $k_H(u,t;u',t')$ onto $\phi_n(u',t')$ is denoted by $ c_n(u,t)$ and is given by,
\begin{equation}
    c_n(u,t) =  \iint k_H(u,t;u',t') \phi_n(u',t') ~du' ~dt'
\end{equation}

Using the above, (28) is expressed as,
\begin{align}
\label{eq:obj_trans}
     & ||s(u,t) - Hx(u,t)||^2 = ||s(u,t) - \sum_n^ \infty x_n c_n(u,t)||^2
\end{align}

Let $\epsilon (x) {=} ||s(u,t) - \sum_n^ \infty x_n \phi_n(u,t)||^2$. Then its expansion is given by,
\begin{align}
\label{eq:ep}
     & \epsilon (x) = \langle s(u,t),s(u,t) \rangle - 2\sum_n ^ \infty  x_n \langle c_n(u,t),s(u,t) \rangle \\
     & + \sum_n^ \infty x_n^2 \langle c_n(u,t), c_n(u,t) \rangle \nonumber + \sum_n^ \infty \sum_{n' \neq n}^ \infty x_n x_{n'}  \langle c_n(u,t), c_{n'}(u,t) \rangle
\end{align}

Then the solution to achieve minimal $\epsilon(x)$ is obtained by solving for $\pdv{\epsilon(x)}{x_n} = 0$ as in \eqref{eq:solution}.
\begin{align}
\label{eq:solution}
    x_n^{opt} & {=}  \frac{\langle s(u,t), c_n(u,t) \rangle + \sum_{n'\neq n}^ \infty x_{n'} \langle c_{n'}(u,t), c_n(u,t) \rangle }{\langle c_n(u,t), c_n(u,t) \rangle}
\end{align}
where $\langle a(u,t), b(u,t) \rangle {=} \iint a(u,t) b^*(u,t) ~du ~dt$ denotes the inner product. 
Let $\langle c_{n'}(u,t), c_n(u,t) \rangle = 0$, i.e., the projections $\{ c_n(u,t)\}_n$ are orthogonal basis. Then we have a closed form expression for $x^{opt}$ as in \eqref{eq:x_opt}.
\begin{align}
\label{eq:x_opt}
    x_n^{opt} & {=}  \frac{\langle s(u,t), c_n(u,t) \rangle}{\langle c_n(u,t), c_n(u,t) \rangle}
\end{align}

Substitute \eqref{eq:x_opt} in \eqref{eq:ep}, it is straightforward to show that $\epsilon(x){=} 0$.
\end{proof}

\subsection{Proof of Theorem 2: Eigenfunction Precoding}
\label{app:thm_2}
\begin{proof}
The 4-D kernel $k_H(u,t;u',t')$ is decomposed into two separate sets of eigenfunction $\{\phi_n(u',t')\}$ and $\{\psi_n(u, t) \}$ using Theorem 1 as in (30). By transmitting the conjugate of the eigenfunctions, $\phi_n(u',t')$ through the channel $H$, we have that,  
\begin{align}
   & H \phi_n^*(u',t') = \iint k_H(u,t;u',t') \phi_n^*(u',t') ~du' ~d t' \nonumber \\ 
   & {=} \iint \sum_{n}^ \infty \{\sigma_n \psi_n(u,t) \phi_n(u',t')\} \phi_n^*(u',t') ~d t' ~d f' \nonumber \\ 
   & {=} \iint \sigma_n \psi_n(u,t) |\phi_n(u',t')|^2 \nonumber\\
   & + \sum_{n'\neq n}^ \infty \sigma_{n'} \psi_{n'}(u,t) \phi_{n'}(u',t')\ \phi_n^*(u',t') ~du' ~d t' \nonumber \\
   & {=} \sigma_n \psi_n(u,t)
\end{align}
where $\psi_n(u,t)$ is also a 2-D eigenfunction with the orthogonal property as in (31). 

From Lemma 2, if the set of projections, $\{c_n(u,t)\}$ is the set of eigenfunctions, $\{\psi_n(u,t)\}$, which has the above orthogonal property, we achieve the optimal solution as in \eqref{eq:x_opt}. Therefore, let $x(u,t)$ be the linear combination of $\{\phi_n^*(u,t)\}$ with coefficients $\{x_n\}$ as in \eqref{eq:construct},

\begin{equation}
\label{eq:construct}
    x(u,t) = \sum_n^ \infty x_n \phi_n^*(u,t) 
\end{equation}

Then \eqref{eq:obj_trans} is rewritten as in \eqref{eq:obj_trans2},
\begin{align}
\label{eq:obj_trans2}
     & ||s(u,t) - Hx(u,t)||^2 = ||s(u,t) - \sum_n^ \infty x_n \sigma_n \psi(u,t)||^2
\end{align}
   
Therefore, optimal $x_n$ in \eqref{eq:x_opt} is obtained as in \eqref{eq:opt},

\begin{equation}
\label{eq:opt}
    x_n^{opt} = \frac{\langle s(u,t), \psi_n(u,t) \rangle}{\sigma_n} 
\end{equation}

Substituting \eqref{eq:opt} in \eqref{eq:construct}, the transmit signal is given by \eqref{eq:x_opt2},
\begin{equation}
\label{eq:x_opt2}
    x(u,t) = \sum_n^ \infty \frac{\langle s(u,t), \psi_n(u,t) \rangle}{\sigma_n} \phi_n^*(u,t). 
\end{equation}
\end{proof}

\subsection{Proof of Corollary 1}
\label{app:EP_space}
\begin{proof}
First we substitute the 4-D kernel $k_H(u,t;u',t')$ with the 2-D kernel $k_H(u,u')$ in Theorem 2 which is then decomposed by the 2-D HOGMT. Then following similar steps as in Appendix~\ref{app:thm_2} it is straightforward to show (34).
\end{proof}

\section{Results on Interference}
\label{App:results_interference}
\begin{figure}[h]
  \centering
  \includegraphics[width=1\linewidth]{figures/hst_1_10_50_100.pdf}
  \caption{Kernel $k_H(u,t;u',t')$ for $u {=} 1$ at a) $t {=} 1$, b) $t {=} 10 $, c) $t {=} 50$ and d) $t {=} 100$.}
  \label{fig:hst_1_10_50_100}
\label{fig:hst_1_10_50_100}
\end{figure}
Figure~\ref{fig:hst_1_10_50_100} shows the channel response for user $u {=} 1$ at $t{=}1$, $t{=}10$, $t{=}50$ and $t{=}100$, where at each instance, the response for user $u {=} 1$ is not only affected by its own delay and other users' spatial interference, but also affected by other users' delayed symbols. 
This is the cause of joint space-time interference which necessitates joint precoding in the 2-dimensional space using eigenfunctions that are jointly orthogonal.

\section{Proof of Theorem 1: Correctness of AE for KLT}
Consider the AE loss function $\mathcal{L}{=}\mathcal{J}{+}\Omega$, where $\mathcal{J}{=}\|\textbf{x}-\hat{\textbf{x}}\|^2$ ($\textbf{x}$ represents each column of the hankel matrix $\textbf{U}$ which is used to train the AE) and $\Omega$ as defined below,
\begin{equation}
\Omega(\boldsymbol{W}_{dec})= \lambda \left([\boldsymbol{W}_{dec}^H\boldsymbol{W}_{dec}-\boldsymbol{I}] + \boldsymbol{W}_{dec}^H\mathbf{R_{\mathbf{xx}}}\boldsymbol{W}_{dec}\right)
\label{eq:KLT}
\vspace{-10pt}
\end{equation}
$\mathbf{R_{\mathbf{xx}}}$ is the input correlation matrix, and  $\lambda{>}0$ is a constant used to penalize the deviation from the KLT eigen bases. The first term ensures that the basis vectors are orthonormal and the final term ensures that the transformed signal components are uncorrelated.

\textbf{Proof:} 
We first observe that finding the eigenfunctions in KLT decomposition and reconstruction is equivalent to solving the following objective function: $\min _{\mathbf{\Phi}}\left\|\mathrm{U}-\mathrm{\Phi\Phi}^{T} \mathrm{U}\right\|_{F}^{2} \text { s.t. } \mathrm{\Phi}^{T} \mathrm{\Phi}=\mathrm{I}$, where $U$ is the Hankel matrix.
Given the linear AE is trained by sequentially feeding in each column of $U$ (denoted as $x$), the encoder and decoder is given by $y=W_{enc}x$ and $\hat{x}=W_{dec}y$. Since the cost function $\mathcal{J}$ is the total squared difference between output and input, then training the autoencoder on the input hankel matrix $\textbf{U}$ solves the following: $\min_{\boldsymbol{\Phi}}\left\|\mathrm{U}-\mathrm{W}_{dec} \mathrm{W}_{enc} \mathrm{U}\right\|_{F}^{2}{+}\Omega_{TED}(\boldsymbol{\Phi})$. 
In \cite{NN_PCA} it is shown that when setting the gradients to zero, $W_{enc}$ is the left Moore-Penrose pseudoinverse of $\mathrm{W}_{dec}$): $\mathrm{W}_{enc}=\mathrm{W}_{dec}^{\dagger}=\left(\mathrm{W}_{dec}^{H} \mathrm{~W}_{dec}\right)^{-1} \mathrm{~W}_{dec}^{H}$ and hence the loss function becomes $\min_{\boldsymbol{\Phi}}\left\|\mathrm{U}-\mathrm{W}_{dec} \mathrm{W}_{dec}^{\dagger} \mathrm{U}\right\|_{F}^{2}{+}\Omega_{TED}(\boldsymbol{\Phi})$. The regularization function $\Omega_{TED}(\boldsymbol{\Phi})$ ensures that the weights are orthonormal, i.e., $\mathrm{W}_{dec} \mathrm{W}_{dec}^{H}=\textbf{I}$, and therefore implies that $\mathrm{W}_{enc}{=}\mathrm{W}_{dec}^{H}$. Then minimizing the loss function of the AE is equivalent to minimizing the objective in KLT since the KLT optimization and AE loss functions are equivalent and the eigen decomposition is a unique representation.
Therefore, we have that $\mathrm{W}_{dec}{=}\mathrm{W}_{enc}^H$ is the same as $\Phi$.
This can be extended to show that a nonlinear AE extracts the nonlinear kernel-KLT eigenfunctions.

\textbf{a) Extension to Deep Autoencoders:}
Let $W_1,{\ldots},W_d$ be the weight matrices corresponding to the $d$ layers at the encoder. Then we calculate $W_{enc}$ or $W_{dec}$ from the trained weights after training as, $W_{enc}=W_1W_2{\ldots}W_d=W_{dec}^H$. Then the proof above holds \cite{ghojogh2019unsupervised}.   

\textbf{b) Extension to Nonlinear Deep Autoencoders:}
Let $W_1,{\ldots},W_d$ be the nonlinear functions corresponding to the $d$ layers at the encoder. Then we calculate $W_{enc}$ or $W_{dec}$ from the trained weights and biases (or from the functions they describe) after training as, $W_{enc}(x)=W_d({\ldots}(W_2(W_1(x)){\ldots})$. Then $W_{dec}$ learns a kernel or non-linear eigen transformation \cite{ghojogh2019unsupervised}.

\bibliographystyle{IEEEtran}
\bibliography{references, ref_precoding}